\documentclass[letterpaper,
onecolumn,
superscriptaddress,
showpacs,
nofootinbib,
notitlepage,
unpublished
]
{quantumarticle}
\pdfoutput=1
\usepackage[numbers,sort&compress]{natbib}
\usepackage[utf8]{inputenc}
\usepackage{amsmath, amssymb,dsfont}
\usepackage{verbatim}
\usepackage[usenames,dvipsnames]{xcolor}
\usepackage{graphicx}
\usepackage[pdftex, plainpages=false]{hyperref}
\usepackage{mathtools}
\usepackage{stmaryrd, ifthen}
\usepackage{array}
\usepackage{tabularx,float}
\usepackage[export]{adjustbox}
\usepackage{bm, bbm, yfonts, wasysym}
\usepackage{enumitem}
\usepackage{calc}
\usepackage{etoolbox}
\mathtoolsset{showonlyrefs=false}
\usepackage{footnote}
\usepackage{dcolumn}
\usepackage{microtype}
\usepackage{booktabs}
\usepackage{epstopdf}

\usepackage{tikz}
\usetikzlibrary{shapes.geometric,shapes.misc}
\usetikzlibrary{arrows,matrix,calc,scopes,decorations.markings,snakes}
\usetikzlibrary{arrows.meta}
\usetikzlibrary{intersections, patterns,fit} 
\usetikzlibrary{decorations.pathreplacing,angles,quotes}

\makeatletter
\newsavebox{\@brx}
\newcommand{\llangle}[1][]{\savebox{\@brx}{\(\m@th{#1\langle}\)}%
  \mathopen{\copy\@brx\kern-0.5\wd\@brx\usebox{\@brx}}}
\newcommand{\rrangle}[1][]{\savebox{\@brx}{\(\m@th{#1\rangle}\)}%
  \mathclose{\copy\@brx\kern-0.5\wd\@brx\usebox{\@brx}}}
\makeatother

\hypersetup{%
	bookmarks=true,         
	bookmarksopen=true,
	unicode=true,           
	pdftoolbar=true,        
	pdfmenubar=true,        
	pdffitwindow=false,     
	pdfstartview={FitH},    
	pdftitle={Antilinear superoperator and quantum geometric invariance  for higher-dimensional quantum systems},
	pdfauthor={},
	pdfsubject={},          
	pdfcreator={},          
	pdfproducer={pdfLaTeX}, 
	pdfkeywords={antilinear superoperator} {entanglement} {concurrence} {fidelity},
	pdfnewwindow=true,      
	colorlinks,
	citecolor=BlueViolet,
	filecolor=BlueViolet,
	linkcolor=BlueViolet,
	urlcolor=BlueViolet,
	linktocpage=true
}

\definecolor{nblue}{rgb}{0.2,0.2,0.7}
\definecolor{ngreen}{rgb}{0.2,0.6,0.2}
\definecolor{nred}{rgb}{0.7,0.2,0.2}
\definecolor{nblack}{rgb}{0,0,0}

\usepackage{amsthm}

\theoremstyle{definition}
\newtheorem{definition}{Definition}
\newtheorem{example}[definition]{Example}
\newtheorem{lemma}[definition]{Lemma}
\newtheorem{theorem}[definition]{Theorem}
\theoremstyle{remark}

\newcommand\Tr   {\operatorname{Tr}}
\newcommand{\one}{\mathds{I}}
\newcommand{\orcidd}[1]{\href{https://orcid.org/#1}{\includegraphics[width=8pt]{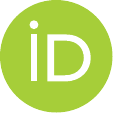}}}

\begin{document}
	
\title{Antilinear superoperator, quantum geometric invariance, and antilinear symmetry for higher-dimensional quantum systems}

\author{Lu Wei}
\email{cox@mail.ustc.edu.cn}
\affiliation{Department of Applied Mathematics \& Statistics, Stony Brook University, Stony Brook, NY 11794, USA}
\affiliation{Department of Computer Science, Stony Brook University, Stony Brook, NY 11794, USA}

\author{Zhian Jia\orcidd{0000-0001-8588-173X}}
\email{giannjia@foxmail.com}
\affiliation{Centre for Quantum Technologies, National University of Singapore, SG 117543, Singapore}
\affiliation{Department of Physics, National University of Singapore, SG 117543, Singapore}

\author{Dagomir Kaszlikowski}
\email{phykd@nus.edu.sg}
\affiliation{Centre for Quantum Technologies, National University of Singapore, SG 117543, Singapore}
\affiliation{Department of Physics, National University of Singapore, SG 117543, Singapore}

\author{Sheng Tan\orcidd{0009-0008-3318-9942}}
\email{tan296@bimsa.cn}
\affiliation{Beijing Institute of Mathematical Sciences and Applications, Beijing, 101408, China}
\affiliation{Yau Mathematical Sciences Center, Tsinghua University, Beijing, 100084, China}
\affiliation{Department of Mathematics, Purdue University, West Lafayette, IN 47907, USA}

\begin{abstract}
\noindent
We present a systematic investigation of antilinear superoperators and their applications in studying open quantum systems, particularly focusing on quantum geometric invariance, entanglement distribution, and symmetry.
We study several crucial classes of antilinear superoperators, including antilinear quantum channels, antilinearly unital superoperators, antiunitary superoperators, and generalized $\Theta$-conjugation.
Using the Bloch representation, we present a systematic investigation of quantum geometric transformations in higher-dimensional quantum systems. By choosing different generalized $\Theta$-conjugations, we obtain various metrics for the space of Bloch space-time vectors, including the Euclidean and Minkowskian metrics. Utilizing these geometric structures, we then investigate the entanglement distribution over a multipartite system constrained by quantum geometric invariance.
The strong and weak antilinear superoperator symmetries of the open quantum system are also discussed. Additionally, Kramers' degeneracy and conserved quantities are examined in detail.
\end{abstract}

\maketitle

\section{Introduction}

With the development of quantum information and quantum computation theory, the investigation of preparation, transformation, and measurement of quantum states becomes
more and more important \cite{preskill1998,Nielsen2010, watrous2018theory}.
According to Wigner theorem~\cite{wigner2012group,bargmann1964note}, 
quantum symmetry operations for a closed quantum system can only be unitary or antiunitary.
It is well known that for an open quantum system, the unitary operations reduce to completely positive trace-preserving (CPTP) maps. 
Therefore it is natural to consider the case corresponding to antiunitary operations. To our knowledge, this has not been systematically investigated so far.
By tracing the environment, the antiunitary operations of the closed system reduce to antilinear CPTP maps, which are special cases of more general antilinear superoperators.

Antilinear operators and superoperators play a crucial role in studying symmetries and transformations of quantum systems.
In the framework of quantum field theory,
charge conjugation symmetry ($\mathsf{C}$-symmetry), time-reversal symmetry ($\mathsf{T}$-symmetry), and the combination of parity symmetry with time-reversal symmetry ($\mathsf{PT}$-symmetry),
are all antilinear operators \cite{peskin2018introduction,sachs1987physics,geru2019time,bender2019pt}.
In quantum information theory, antilinear operators and superoperators both have fundamental and practical importance. For instance, Hill-Wootters conjugation, or more generally, $\Theta$-conjugation \cite{Hill1997entanglement, Uhlmann2000fidelity,mintert2005measures} is a crucial tool for characterizing and quantifying quantum entanglement.
Antilinear EPR transformation and antilinear quantum teleportation transformation can be used to investigate various quantum information protocols (see \cite{uhlmann2016anti} and references therein). 
A systematic investigation of antilinear operators acting on state vectors is presented in \cite{uhlmann2016anti}.
However, in an open quantum system, it is more natural to consider the antilinear superoperators which have not been systematically investigated before and it will be one of the main focuses of this work.

For the simplest case, the qubit system, Bloch representation provides an extremely convenient description of the quantum system \cite{Nielsen2010, preskill1998}, where the $\Theta$-conjugation has a concise representation using the Bloch vectors.
However, generalization of Bloch representation to higher-dimensional cases remains open \cite{Hioe1981N,Jakobczyk2001geometry,Kimura2003bloch}. Two reasons behind this are that in higher-dimensional cases, the set of Bloch vectors for quantum states is no longer a ball and higher-dimensional vectors are more complicated to visualize and manipulate.
In this work, we will discuss the Bloch representation for a higher-dimensional quantum system based on the Hilbert-Schmidt basis of the real vector space consisting of all Hermitian operators. This provides us a good framework to study the generalized $\Theta$-conjugation, which, by definition, is a class of antilinear superoperators.
To this end, we present a systematic investigation of the antilinear superoperators, including antilinear quantum channel, antilinear unital superoperator, antiunitary superoperator and generalized $\Theta$-conjugation.
The generalized $\Theta$-conjugation turns out to be closely related to the geometric transformations of a quantum state.
For qubit case, these geometric transformations and geometric invariance have been extensively explored \cite{Han1997stokes,Han1999wigner,Verstraete2002lorentz,Jaeger2003quantum,Eltschka2015monogamy}. The Lorentz transformation corresponds to stochastic local operation and classical communication (SLOCC) (see, e.g. \cite{li2018stochastic} and references therein). 
For the qudit case, using generalized $\Theta$-conjugation, we present a concise description of Lorentzian and Euclidean invariance in a quantum system. The Lorentzian invariance of quantum states can be used to study entanglement and other properties of multipartite systems \cite{Eltschka2015monogamy,eltschka2018distribution,eltschka2020maximum,aschauer2003local,wyderka2020characterizing}.

On the other hand, understanding the symmetries of open quantum systems has attracted much attention in recent years \cite{buvca2012note,Albert2014symmetries,Lieu2020symmetry,Lieu2020tenfold,Altland2021symmetry,de2022symmetry,McDonald2022exact,Sa2023symmetry,Kawabata2023symmetry}.
While linear symmetries characterized by unitary operators have been extensively investigated, symmetries characterized by linear and antilinear superoperators have received less attention.
A deep understanding of antilinear superoperator symmetry may shed new light on the understanding of generalized symmetries, such as non-invertible symmetries \cite{cordova2022snowmass,brennan2023introduction,mcgreevy2023generalized,luo2023lecture,shao2024whats,SchaferNameki2024ICTP,Bhardwaj2024lecture}, Hopf and weak Hopf symmetries \cite{Bais2002quantum,bais2003hopf,jia2023boundary,Jia2023weak,jia2024weakTube,jia2024generalized}, higher-form symmetries \cite{gaiotto2015generalized,kapustin2017higher,gomes2023introduction,Bhardwaj2024lecture}, and more.
For Hamiltonians without unitary symmetries, they are classified by the tenfold classification based on antiunitary symmetries: time-reversal symmetry $\mathsf{T}$, charge conjugation symmetry $\mathsf{C}$, and their composition, called chiral symmetry $\mathsf{S} = \mathsf{C} \circ \mathsf{T}$ \cite{Altland1997tenfold,Ryu2010tenfold}.
When generalized to open systems, this classification problem will also strongly depend on the antilinear symmetries of open quantum systems \cite{Lieu2020tenfold,Altland2021symmetry,de2022symmetry,Sa2023symmetry,Kawabata2023symmetry}.
Motivated by the aforementioned challenges, in this work, we also provide a detailed discussion of the weak and strong antilinear superoperator symmetries of open quantum systems.

The work is organized as follows. In Sec.~\ref{sec:theta}, we introduce the basic definitions of antilinear superoperators and discuss their representations and properties.  We systematically study antilinearly CPTP maps, antilinear unital superoperators, antiunitary superoperators, and generalized $\Theta$-conjugations.
Then as an application, the geometric representation of generalized $\Theta$-conjugation and its relationship with quantum geometric invariance are discussed in Sec.~\ref{sec:geo}.
In Sec.~\ref{sec:bloch}, we briefly discuss Bloch representation of higher-dimensional quantum states based on Hilbert-Schmidt basis. The generalized $\Theta$-conjugation can be realized as a geometric transformation for Block tensors.
Using quantum geometric invariance of quantum states, we study the distribution of entanglement over a multipartite system in Sec.~\ref{sec:mono}. 
In Sec.~\ref{sec:symmetry} we discuss antilinear superoperator symmetry of open quantum systems, which is also a crucial application of antilinear superoperators. We discuss both weak and strong antilinear symmetries and prove the Kramers' degeneracy given by the antilinear superoperator symmetries.
Finally, in Sec.~\ref{sec:conclusion}, we give some concluding remarks and discussions.

\section{Antilinear superoperators}
\label{sec:theta}

The antilinearity has many applications in physics since Wigner observed that time-reversal operation in quantum mechanics is characterized by antiunitary operators \cite{wigner2012group,bargmann1964note}.
In quantum information theory, Hill-Wootters conjugation  \cite{Hill1997entanglement} and
Uhlmann's generalization of $\Theta$-conjugation  \cite{Uhlmann2000fidelity} are crucial examples of antilinear operators which play a key role in investigating quantum correlations.
Some fragmentary discussions of antilinear operators are presented in Refs.~\cite{paulsen2002completely,uhlmann2016anti}.
However, a systematic investigation of antilinearity, especially antilinear superoperators and antilinear quantum effects, is still largely lacking.
In this section and Appendix~\ref{app:AntiSup}, we present a careful discussion of antilinear superoperators, for which antilinear operators are special cases (the antilinear superoperator with Kraus rank equal to one). Various special classes of antilinear superoperators will be discussed, and these antilinear superoperators may have more potential applications in investigating quantum information theory.  Particular applications in the study of geometric invariance and distribution of quantum correlations for higher-dimensional quantum systems will be discussed later in this work.

\subsection{Representations of antilinear superoperators}

To begin with, let us introduce the basic definitions related to the antilinear (or conjugate linear) superoperators, which are natural generalizations of their linear counterparts. In what follows, $\mathbf{B}(\mathcal{X})$ denotes the operator space associated to a finite-dimensional system $\mathcal{X}$.

\begin{definition}
	Let $\mathcal{M}:\mathbf{B}(\mathcal{X}) \to  \mathbf{B}(\mathcal{Y})$ be a mapping between Banach spaces  $\mathbf{B}(\mathcal{X})$ and $\mathbf{B}(\mathcal{Y})$.
	It is called an antilinear superoperator if
	\begin{equation}
		\mathcal{M}(\alpha \rho_1 +\beta \rho_2)=\alpha^*\mathcal{M}( \rho_1) +\beta^* \mathcal{M}( \rho_2)
	\end{equation}
	for all $\alpha,\beta \in \mathbb{C}$ and $\rho_1,\rho_2 \in \mathbf{B}(\mathcal{X})$. 
	Since $\mathbf{B}(\mathcal{X})$ and $\mathbf{B}(\mathcal{Y})$ are both inner-product spaces with Hilbert-Schmidt inner product,
	we can introduce the Hermitian adjoint superoperator $\mathcal{M}^{\dagger}$ by
	\begin{equation}
		\langle \mathcal{M}^{\dagger} (\sigma) ,\rho\rangle_{\mathbf{B}(\mathcal{X})}
		=\langle \sigma, \mathcal{M} (\rho) \rangle_{\mathbf{B}(\mathcal{Y})},
	\end{equation}
	and similarly, we can introduce the antilinear Hermitian adjoint superoperator $\mathcal{M}^{\ddag}$ by
	\begin{equation}
		\overline{	\langle \mathcal{M}^{\ddag} (\sigma),\rho \rangle}_{\mathbf{B}(\mathcal{X})}= \langle \sigma, \mathcal{M}(\rho)\rangle_{\mathbf{B}(\mathcal{Y})},
	\end{equation}
	for all  $\rho \in \mathbf{B}(\mathcal{X})$, and $\sigma \in \mathbf{B}(\mathcal{Y})$.
	The set of antilinear superoperators forms a linear space, which we denote as  $\mathbf{B}^{(2)}_{\mathrm{anti}}(\mathcal{X} , \mathcal{Y})$; when $\mathcal{X}=\mathcal{Y}$, we will simply denote it as $\mathbf{B}^{(2)}_{\mathrm{anti}}(\mathcal{X} )$.\,\footnote {The notation here has a categorical meaning: the morphisms between two Hilbert spaces are called 1-morphisms (operators), and the morphisms between 1-morphism spaces are called 2-morphisms (superoperators). See Appendix~\ref{app:AntiSup}.}
\end{definition}

The composition of an antilinear superoperator with a linear superoperator is antilinear, while the composition of an antilinear superoperator with another antilinear superoperator is linear.
This means that different from the space of all linear superoperators $\mathbf{B}^{(2)}(\mathcal{X} )$ which is an algebra, $\mathbf{B}^{(2)}_{\mathrm{anti}}(\mathcal{X} )$ is not an algebra, because the composition is not closed.
However, one can show that  $\mathbf{B}^{(2)}_{\mathrm{anti}}(\mathcal{X} )$ is a bimodule over $\mathbf{B}^{(2)}(\mathcal{X} )$. 
For antilinear $\mathcal{M}$, it is easy to show that the Hermitian adjoint $\mathcal{M}^{\dagger}$ becomes linear, but the antilinear Hermitian adjoint  $\mathcal{M}^{\ddag}$ is still antilinear. We also have $(\mathcal{M}^{\ddag})^{\ddag}=\mathcal{M}$, and ($\mathcal{M}\circ \mathcal{N})^{\ddag} =\mathcal{N}^{\ddag}  \circ \mathcal{M}^{\ddag}$. 
For linear $\mathcal{N}$ and antilinear $\mathcal{M}$, we have $(\mathcal{M}\circ \mathcal{N})^{\ddag}=\mathcal{N}^{\dagger}\circ \mathcal{M}^{\ddag}$, and
$(\mathcal{N}\circ \mathcal{M})^{\ddag}=\mathcal{M}^{\ddag}\circ \mathcal{N}^{\dagger}$.
An antilinear superoperator $\mathcal{M}$ is called Hermitian if $ \mathcal{M}^{\ddagger}=\mathcal{M}$, and skew Hermitian if $\mathcal{M}^{\ddagger}=-\mathcal{M}$.


Tensor product between two antilinear superoperators $\mathcal{M}$ and $\mathcal{N}$ is well-defined in the way that $\mathcal{M}\otimes \mathcal{N} (\rho \otimes \sigma)= \mathcal{M}(\rho) \otimes \mathcal{N}(\sigma)$. However, the tensor product between linear and antilinear superoperators cannot be consistently defined in this way. In particular, we cannot define the tensor product between antilinear $\mathcal{M}$ and (linear) identity superoperator $\mathcal{I}$.
This means that antilinearity is a global property of the quantum system, making it a crucial tool for characterizing quantum correlations.

One of the most crucial examples of an antilinear superoperator is the complex conjugation $\mathcal{K}$ in a particular basis, given by $\mathcal{K} (\rho) =\rho^{*}$.
It can be checked that $\mathcal{K}^{-1}=\mathcal{K}^{\ddag}=\mathcal{K}$. 
As we will show later, complex conjugation superoperator plays a key role in characterizing antilinear superoperators.  
The following result is the main tool that we will utilize to investigate antilinear superoperators.

\begin{theorem}
	\label{lemma:antilinear}
	Any antilinear superoperator $\mathcal{M}$ can be decomposed as $\mathcal{M}= \mathcal{M}_{L} \circ \mathcal{K} = \mathcal{K} \circ \mathcal{M}_L^*$, with $\mathcal{M}_L$ and $\mathcal{M}_L^*$ being linear superoperators, called the left and right linearizations of $\mathcal{M}$, respectively. The decomposition is referred to as the fundamental decomposition of an antilinear superoperator.
\end{theorem}

\begin{proof}
	By choosing a basis $\{E_{ij}=|i\rangle \langle j|\}$, notice that the complex conjugation of an operator in this basis has the property: $\mathcal{K}(E_{ij})=E_{ij}$ and $\mathcal{K}(\rho)=\mathcal{K} (\sum_{ij} \rho_{ij} E_{ij})=\sum_{ij} \rho_{ij}^* E_{ij}=\rho^*$.  We can define a linear superoperator $\mathcal{M}_L$ such that $\mathcal{M}_L(E_{ij}):=\mathcal{M} (E_{ij})$. Then we see that $ \mathcal{M} (\rho) =\sum_{ij} \rho^*_{ij}\mathcal{M}_L(E_{ij}) = \mathcal{M}_L \circ \mathcal{K} (\rho)$ for all $\rho$, which shows that $\mathcal{M}= \mathcal{M}_L \circ \mathcal{K}$.
	We can also define a linear operator $\mathcal{M}_L^*$ such that $\mathcal{M}_L^* (E_{ij}) =\mathcal{M}(E_{ij})^*$. In this way,
	$\mathcal{M}(\rho)=\sum_{ij} \rho_{ij}^* \mathcal{M}(E_{ij}) =\sum_{ij} \rho_{ij}^* [\mathcal{M}_L^*(E_{ij})]^*=\mathcal{K} (\sum_{ij} \rho_{ij} \mathcal{M}_L^*(E_{ij}))=\mathcal{K}\circ \mathcal{M}_L^*(\rho)$ for all $\rho$. This completes the proof.
\end{proof}

With the same idea we can show that if $A:\mathcal{X}\to \mathcal{Y}$ is an antilinear operator, then there exist corresponding linear operators $A_L$ and $A_L^*$ such that $A=A_L K=K A_L^*$ where $K$ is the complex conjugate operator in the standard basis.
Notice that we assume that all antilinear superoperators act from left to right. When considering the right action, we will always use its linearization.
For antilinear operator $A$,
there is a corresponding antilinear superoperator $\mathcal{M} (\rho) =A_L \mathcal{K} (\rho) A_L^{\dagger}$. Thus, the antilinear operator can be regarded as a special case of an antilinear superoperator.

Now let us discuss the representation of antilinear superoperators.
The simplest one is the natural representation. By introducing a vector map
\begin{equation}
\| \bullet \rrangle :  |i\rangle \langle j| \mapsto |i\rangle |j\rangle,
\end{equation}
it is clear that the mapping
\begin{equation}
	N(\mathcal{M}): \| \rho\rrangle  \mapsto \|  \mathcal{M} (\rho)\rrangle 
\end{equation}
is antilinear (since it is the composition of a linear map and an antilinear map). Thus, we obtain an antilinear operator representation of $\mathcal{M}$. From theorem~\ref{lemma:antilinear} we see that $N(\mathcal{M}_L)$ coincides with the linearization $N(\mathcal{M})_L$, 
\begin{equation}
	N(\mathcal{M})_L=	\sum_{i,j}\sum_{k,l} \langle E_{k,l}, \mathcal{M}_L (E_{i,j}) \rangle E_{k,i}\otimes E_{l,j}.
\end{equation}
The antilinear operator $N(\mathcal{M})_L K$ will be called the natural representation of $\mathcal{M}$.
The natural representation for the antilinear Hermitian conjugate $\mathcal{M}^{\ddag}$ is of the form
\begin{equation}
	N(\mathcal{M}^{\ddag} )_L=	\sum_{i,j}\sum_{k,l} \langle E_{i,j}, \mathcal{M}_L (E_{k,l}) \rangle E_{k,i}\otimes E_{l,j},
\end{equation}
which means that $N(\mathcal{M}^{\ddag} )=N(\mathcal{M})^{\ddag}$.
It is easy to verify that $\mathcal{M}^{\ddag}\circ \mathcal{M}=\mathcal{I}$ if and only if $N(\mathcal{M}^{\ddag}) N(\mathcal{M})=\one$.

From theorem~\ref{lemma:antilinear}, we see that an antilinear superoperator $\mathcal{M}$ can be linearized as 
\begin{equation}
    \mathcal{M}=\mathcal{M}_L\circ \mathcal{K}
\end{equation}
with $\mathcal{K}$ the complex conjugation superoperator and $\mathcal{M}_L$ a linear superoperator. Suppose that   $\{A_j\}$ and $\{B_j\}$ are the Kraus operators for $\mathcal{M}_L$, then we obtain the Kraus decomposition for $\mathcal{M}$: 
\begin{equation} \label{eq:kraus}
	\mathcal{M} (\rho)= \mathcal{M}_L\circ \mathcal{K}(\rho) = \sum_j A_j \rho^* B_j^{\dagger}.
\end{equation}
The minimal number of terms $A_j,B_j$ that appear in the Kraus decomposition is called the Kraus rank of $\mathcal{M}$.
The Kraus representation for $\mathcal{M}^{\ddag}$ is thus
\begin{equation}
	\mathcal{M}^{\ddag}(\rho)=\sum_j A_j^{T} \rho^* B_j^*.
\end{equation}
Kraus representations exist for all antilinear superoperators, but they are not unique in general.

Another representation is the Choi-Jamio{\l}kowski representation, which is a useful representation for characterizing the positivity of the superoperators. For $\mathcal{M} \in \mathbf{B}^{(2)}_{\mathrm{anti}} (\mathcal{X},\mathcal{Y} )$, we have 
\begin{equation}
	J(\mathcal{M}) =(\mathcal{M}_L\otimes \mathcal{I}_{\mathcal{X} } ) (|\Omega\rangle \langle \Omega| ) = \sum_{i,j} \mathcal{M}_L(E_{ij}) \otimes E_{ij},
\end{equation}
where $|\Omega\rangle=\sum_i |i\rangle |i\rangle \in \mathcal{X}\otimes \mathcal{X}$ .   The operator $J(\mathcal{M}) $ is called the Choi-Jamio{\l}kowski representation of $\mathcal{M}$. It is easy to verify that 
\begin{equation}
	\mathcal{M}(\rho)=\operatorname{Tr}_{\mathcal{X}} (J(\mathcal{M})   (\mathds{I}_{\mathcal{Y}} \otimes \rho^{\dagger} )   ).
\end{equation}

From an open-system's point of view, a superoperator is a quantum operation obtained by partially tracing the environment of a closed system. This also works for antilinear superoperator $\mathcal{M}$, where the resulting representation is called the Stinespring representation.
Suppose that $U,V \in \mathbf{B}(\mathcal{X},\mathcal{Y}\otimes \mathcal{Z})$ are Stinespring representation operators of $\mathcal{M}_L$, then from theorem~\ref{lemma:antilinear} we have
\begin{equation} \label{eq:stinespring}
	\mathcal{M}(\rho) =\operatorname{Tr}_{\mathcal{Z}} (U\rho^* V^{\dagger}).
\end{equation}

The Kraus representation is usually regarded as the most fundamental representation. In the following result, we provide some explicit formulae to translate the Kraus representation into other representations.

\begin{lemma}\label{lemma:relation}
	Suppose $\mathcal{M}\in \mathbf{B}^{(2)}_{\mathrm{anti}} (\mathcal{X},\mathcal{Y} )$ is an antilinear superoperator and its Kraus representation is given by Eq.~(\ref{eq:kraus}), then we have:
	\begin{enumerate}
		\item The natural representation of $\mathcal{M}$ is $N(\mathcal{M})_L= \sum_j A_j\otimes B_j^*$.
		\item The Choi-Jamio{\l}kowski representation is  $ J(\mathcal{M}) =\sum_j \|A_j\rrangle \llangle B_j\| $.
		\item The Stinespring dilation is given by Eq.~(\ref{eq:stinespring})
		with $U=\sum_j A_j \otimes e_j$ and $V=\sum_j B_j\otimes e_j$, where $\{e_j\}$ is an orthonormal basis of $\mathcal{Z}$.
	\end{enumerate}
\end{lemma}

\begin{proof}
	These claims can be verified by straightforward calculation.
\end{proof}

\subsection{Antilinear quantum channel}
We have provided several ways to represent an antilinear superoperator.
Let us now consider the antilinear channel, which is a natural generalization of (linear) quantum channel.
\begin{definition}
	The following are some crucial classes of antilinear superoperators:
	\begin{itemize}
		\item $\mathcal{M}$ is called antilinearly trace-preserving (TP) if $\Tr (\mathcal{M}(\rho)) =(\Tr (\rho))^*$ for all $\rho \in \mathbf{B}(\mathcal{X})$.
		
		\item
		$\mathcal{M}$ is called antilinearly completely positive (CP) if it is positive, i.e., it maps positive semidefinite operators to positive semidefinite operators, and $\mathcal{K} \otimes \mathcal{M}$ is positive. The antilinear CP superoperators are also called antilinear quantum operations.
		\item The antilinearly  CPTP superoperators are called antilinear quantum channels. 
	\end{itemize}
\end{definition}

From theorem~\ref{lemma:antilinear}, we have the following characterizations of an antilinear quantum channel.

\begin{theorem} \label{theorem:CPTP}
	Suppose the antilinear superoperator $\mathcal{M}$ has the decomposition	$\mathcal{M}=\mathcal{M}_L\circ \mathcal{K} $. Then the following hold: 
	\begin{enumerate}
		\item $\mathcal{M}$ is antilinearly TP if and only if $\mathcal{M}_L$ is TP.
		
		\item $\mathcal{M}$ is antilinearly CP if and only if $\mathcal{M}_L$ is CP.
		
		\item   $\mathcal{M}$ is an antilinear quantum channel if and only if $\mathcal{M}_L$ is a linear quantum channel.	
	\end{enumerate}
\end{theorem}

\begin{proof}
	(1)	For sufficiency, suppose that $\mathcal{M}_L$ is a TP superoperator. Then it is easy to check that $\mathcal{M}=\mathcal{M}_L\circ \mathcal{K}$ is an antilinear TP superoperator. 
	For necessity, using the Kraus decomposition $\mathcal{M} (\rho)=\sum_j A_j \rho^* B_j^{\dagger}$, since $\Tr ( \mathcal{M} (\rho)) =\Tr (\rho^*)$ for all $\rho^*$, the map $\sum_j A_j  (\cdot) B_j^{\dagger}$ is TP map; this map is precisely $\mathcal{M}_L$.
	
	(2) Notice that $\rho^*$ is positive semidefinite if $\rho$ is positive semidefinite. Then $\mathcal{M}_L$ being CP implies directly that $\mathcal{M}$ is CP. For the other direction, consider the Kraus decomposition $\mathcal{M} (\rho)=\sum_j A_j \rho^* B_j^{\dagger}$. The CP condition implies that $B_j=A_j$. In this way, $\mathcal{M} (\rho)=\sum_j A_j \rho^* A_j^{\dagger}$ for all $\rho^*$, which further implies that $\mathcal{M}_L$ is CP. More concisely, from the fact that $\mathcal{K}\otimes \mathcal{M}=(\mathcal{I} \otimes \mathcal{M}_L) \circ  (\mathcal{K}\otimes \mathcal{K})$ is positive, we obtain that $\mathcal{I}\otimes \mathcal{M}_L$ is positive.
	
	(3) This is a direct result of points (1) and (2).
\end{proof}

In the case of (linear) superoperators, the various representations we have discussed above can be used to characterize the antilinearly CP and TP superoperators. From theorem~\ref{theorem:CPTP}, this becomes a straightforward generalization.

\begin{lemma} \label{lemma:CP}
	For $\mathcal{M} \in \mathbf{B}^{(2)}_{\mathrm{anti}} (\mathcal{X},\mathcal{Y} )$, the following statements are equivalent:
	\begin{enumerate}
		\item $\mathcal{M}$ is antilinearly CP superoperator.
		\item In the Kraus representation Eq.~(\ref{eq:kraus}) of $\mathcal{M}$, $A_j=B_j$, thus
		\begin{equation}
			\mathcal{M}(\rho) =\sum_j A_j \rho^* A_j^{\dagger}.
		\end{equation}
		\item The Choi-Jamio{\l}kowski representation $J(\mathcal{M})$ is a positive semidefinite operator.
		\item In the Stinespring representation Eq.~(\ref{eq:stinespring}) of $\mathcal{M}$, $U=V$, thus 
		\begin{equation}
			\mathcal{M}(\rho) =\operatorname{Tr}_{\mathcal{Z}} (U \rho^* U^{\dagger}).
		\end{equation}
	\end{enumerate}
\end{lemma}

\begin{lemma}\label{lemma:TP}
	For $\mathcal{M} \in \mathbf{B}^{(2)}_{\mathrm{anti}} (\mathcal{X},\mathcal{Y} )$, the following statements are equivalent:
	\begin{enumerate}
		\item $\mathcal{M}$ is antilinearly TP superoperator.
		\item  The Kraus representation in Eq.~(\ref{eq:kraus}) satisfies 
		\begin{equation}
			\sum_j A_j^{\dagger} B_j=\mathds{I}_{\mathcal{X}}.
		\end{equation}
		\item The Choi-Jamio{\l}kowski representation $J(\mathcal{M})$ satisfies
		\begin{equation}
			\operatorname{Tr}_{\mathcal{Y}} J(\mathcal{M}) =\mathds{I}_{\mathcal{X}}.
		\end{equation}
		\item For the Stinespring representation as in Eq.~(\ref{eq:stinespring}), 
		the operators $U,V$ satisfy $U^{\dagger} V=\mathds{I}_{\mathcal{X}}$. 
	\end{enumerate}
\end{lemma}

Using the relations between different representations of $\mathcal{M}$ presented in lemma~\ref{lemma:relation}, the proofs of the above two lemmas are straightforward. Combining the lemma~\ref{lemma:CP} and lemma~\ref{lemma:TP}, we obtain a complete characterization of antilinear quantum channels. 

From the Stinespring representation of the antilinear quantum channel, every antilinear channel can be implemented by performing a joint antiunitary transformation on the system and ancilla and then tracing over the ancilla.

\begin{example}
    Let us now consider some examples of antilinear quantum channels.
    \begin{enumerate}
        \item Antilinear bit flip channel.  The Kraus operators are $A_0=\sqrt{p}\,\one$, $A_1=\sqrt{1-p}\,\sigma_x$ with $p\in [0,1]$, the channel is of the form
        \begin{equation}\label{eq:bitFlip}
            \mathcal{E}(\rho)=\sum_i A_i \mathcal{K}(\rho) A_i^{\dagger},
        \end{equation}
        where $\rho=(\one+\vec{r}\cdot \vec{\sigma} )/2$ with $\vec{r}=(r_x,r_y,r_z)$ the Bloch vector, and $\mathcal{K}$ is the complex conjugation channel.
        The output state $\omega=\mathcal{E}(\rho)$ has a Bloch vector
        \begin{equation}
            \vec{s}= (r_x, (1-2p)r_y, (2p-1)r_z).
        \end{equation}
        This means that states on the $x$-axis are left alone, but states on the $yz$-plane are contracted.
        
        \item Antilinear phase flip channel. The Kraus operators are $A_0=\sqrt{p}\,\one$, $A_1=\sqrt{1-p}\,\sigma_z$ with $p\in [0,1]$, the channel takes a similar form as Eq.~\eqref{eq:bitFlip}. The Bloch vector for the output state reads
        \begin{equation}
            \vec{s}=((2p-1)r_x,(1-2p)r_y,r_z).
        \end{equation}
        Thus states on the $z$-axis are left alone, but states on the $xy$-plane are contracted.
        
         \item Antilinear bit-phase flip channel.  The Kraus operators are $A_0=\sqrt{p}\,\one$, $A_1=\sqrt{1-p}\,\sigma_y$ with $p\in [0,1]$. The Bloch vector for the output state is
         \begin{equation}
             \vec{s}=((2p-1)r_x,-r_y,(2p-1)r_z).
         \end{equation}
         Unlike the linear bit-phase flip channel, the antilinear bit-phase flip channel also flips the states in the $y$-axis, and the states on the $xz$-plane are contracted.
         \item Antilinear depolarizing channel. For input qubit state $\rho$, the channel is defined as
         \begin{equation}
             \mathcal{E}(\rho)=p\,\frac{\one}{2}+ (1-p)\mathcal{K}(\rho).
         \end{equation}
         The corresponding Kraus operators are
         \begin{equation}
             A_0=\sqrt{\frac{4-3p}{4}}\,\one, A_1=\sqrt{\frac{p}{2}}\,\sigma_x, A_2=\sqrt{\frac{p}{2}}\,\sigma_y, A_3=\sqrt{\frac{p}{2}}\,\sigma_z,
         \end{equation}
         the channel takes a similar form as Eq.~\eqref{eq:bitFlip}.
         The Bloch vector for the output state is
          \begin{equation}
             \vec{s}=((1-p)r_x,-(1-p)r_y,(1-p)r_z).
         \end{equation}
       There is a reflection of the Bloch vector along the $y$-axis before contraction.
         \item Antilinear amplitude damping channel. The Kraus operators are
         \begin{equation}
             A_1=\left(\begin{matrix}
                 1 & 0\\
                 0& \sqrt{1-\gamma}
             \end{matrix}\right), \quad A_2 =\left
(\begin{matrix}
                 0& \sqrt{\gamma}\\
                 0&0
             \end{matrix}\right).
         \end{equation}
         The Bloch vector for the output state is
                 \begin{equation}
             \vec{s}=(\sqrt{1-\gamma}r_x,-\sqrt{1-\gamma}r_y,\gamma+(1-\gamma)r_z).
         \end{equation} 
        \item Antilinear phase damping channel. The Kraus operators are
         \begin{equation}
             A_1=\left(\begin{matrix}
                 1 & 0\\
                 0& \sqrt{1-\gamma}
             \end{matrix}\right), \quad A_2 =\left
(\begin{matrix}
                 0& 0 \\
                 0&\sqrt{\gamma}
             \end{matrix}\right).
         \end{equation}
          The Bloch vector for the output state is
                 \begin{equation}
             \vec{s}=(\sqrt{1-\gamma}r_x,-\sqrt{1-\gamma}r_y,r_z).
         \end{equation} 
        Notice that this is in the same form as the antilinear phase flip channel if we set $2p-1=\sqrt{1-\gamma}$. This equivalence arises because the antilinear phase damping channel can be transformed into a phase flip channel by this substitution.
    \end{enumerate}
\end{example}

\subsection{Antilinearly unitary and unital superoperator}
\label{sec:unital}
The unital superoperators have broad applications in quantum information theory \cite{mendl2009unital,watrous2018theory}.  
The unital channels are also called doubly stochastic quantum channels.
In this subsection, we will study their antilinear counterparts.

\begin{definition}
	Let $\mathcal{M}\in \mathbf{B}_{\mathrm{anti}}^{(2)} (\mathcal{X}, \mathcal{Y})$ be an antilinear superoperator.
	It is called antilinearly unital (stochastic) if 
	$\mathcal{M} ( \mathds{I}_{\mathcal{X}}) =\mathds{I}_{ \mathcal{Y} }$. 
	An antilinearly unital and TP superoperator is called antilinearly doubly stochastic superoperator.
\end{definition}

\begin{definition}\label{def:antiunitary}
	Let $\mathcal{M}\in \mathbf{B}_{\mathrm{anti}}^{(2)} (\mathcal{X}, \mathcal{Y})$ be an antilinear superoperator.
	It is called antilinearly weak unitary (or weak antiunitary) if $\mathcal{M}^{\ddagger}=\mathcal{M}^{-1}$; it is called antilinearly (strong) unitary if there exist unitary operators $U,V$ such that $\mathcal{M}(\rho)=U\rho^* V^{\dagger}$.
\end{definition}
For linear superoperators, we have similar definitions for strong and weak unitarity. 
In quantum information literature, the unitary superoperator is what we call strong unitary. Hereinafter, we will call strong (anti)unitary superoperator simply (anti)unitary superoperator whenever there is no ambiguity. A strong (anti)unitary superoperator is always a weak (anti)unitary superoperator, but the reverse is not true.
The natural representation of a strong (anti)unitary superoperator always has a tensor product structure $N(\mathcal{M})_L=U\otimes V^*$. But for general weak unitary superoperator, there is no such a structure.

For maps between $\mathbf{B}(\mathcal{X})$ and $\mathbf{B} (\mathcal{Y})$, we consider those that preserve the norm of the inner product, i.e.,
\begin{equation}
	|\langle \mathcal{M} (\sigma), \mathcal{M}(\rho) \rangle|=|\langle \sigma, \rho\rangle|.
\end{equation}
Inspired by the Wigner theorem, we have only two possible classes of such maps: 
(i) for linear case $\mathcal{M}^{\dagger}=\mathcal{M}^{-1}$;
(ii) for antilinear case $\mathcal{M}^{\ddag} =\mathcal{M}^{-1}$.

\begin{theorem}[Wigner]
    The superoperator transformations between $\mathbf{B}(\mathcal{X})$ and $\mathbf{B} (\mathcal{Y})$ that preserve the norm induced by the Hilbert-Schmidt inner product can only be unitary or antiunitary.
\end{theorem}

The proof is straightforward using the natural representation of linear and antilinear superoperators.

\begin{theorem}
	$\mathcal{M}$ is antilinearly unital if and only if $\mathcal{M}_L$ is unital;
	$\mathcal{M}$ is antiunitary if and only if $\mathcal{M}_L$ is unitary.
\end{theorem}

Antiunitary quantum channels are automatically unital. 
The antilinearly unital superoperators are closely related to antilinearly TP superoperators.
\begin{theorem} \label{thm:TPC}
	$\mathcal{M}$ is antilinearly TP if and only if $\mathcal{M}_L^{\dagger}$ is unital.
\end{theorem}
\begin{proof}
	From theorem~\ref{theorem:CPTP}, $\mathcal{M}$ is antilinearly TP if and only if $\mathcal{M}_L$ is TP, and this is further equivalent to that $\mathcal{M}_L^{\dagger}$ is unital.
\end{proof}

It is clear that the antiunitary superoperator is antilinearly unital. We can also introduce the mixture of antiunitaries 
\begin{equation}
	\mathcal{U} (\rho) =\sum_j p(j) U_j \rho^* U_j^{\dagger},
\end{equation}
where $p(j)$ is a probability distribution and $U_j$'s are a collection of unitary operators. 
Antilinearly Weyl-covariant channel is also a crucial example of unital superoperator,
\begin{equation}
	\mathcal{W}(\rho) =\sum_{i,j\in \mathbb{Z}_N} p(i,j) W_{ij} \rho^* W_{ij}^{\dagger},
\end{equation}
where $W_{ij}=X^{i}_N Z^j_N$ and $X_N,Z_N$ are generalized Pauli operators. It can be proved that an antilinearly Weyl-covariant channel is a mixed antilinearly unitary channel. Another crucial class of antilinearly unital superoperators is generalized $\Theta$-conjugation, which is useful for us to investigate the geometric properties of higher-dimensional quantum systems.

\begin{figure}[t]
	\centering
	\includegraphics[width=8cm]{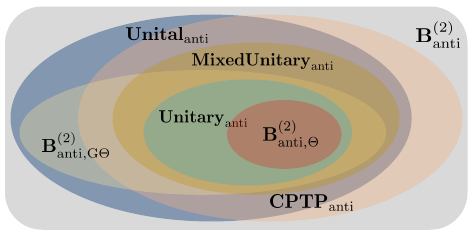}\\
	\caption{A depiction of relations between different classes of antilinear superoperators. Here $\mathbf{B}^{(2)}_{\mathrm{anti}}$ represents the set of all antilinear superoperators;
		$\mathbf{CPTP}_{\mathrm{anti}}$ represents the set of all antilinearly CPTP superoperators; $\mathbf{Unital}_{\mathrm{anti}}$ represents the set of all antilinearly unital superoperators; $\mathbf{MixedUnitary}_{\mathrm{anti}}$ represents the set of all antilinearly mixed-unitary superoperators;
		$\mathbf{Unitary}_{\mathrm{anti}}$ represents the set of all antiunitary superoperators;
		$\mathbf{B}^{(2)}_{\mathrm{anti},\mathrm{G}\Theta}$ represents the set of all generalized $\Theta$-conjugations;
		$\mathbf{B}^{(2)}_{\mathrm{anti},\Theta}$ represents the set of all $\Theta$-conjugations.
	}\label{fig:relation}
\end{figure}

\subsection{Generalized $\Theta$-conjugation}

We now introduce the notion of generalized $\Theta$-conjugation. This is a useful definition for investigating quantum fidelity, quantum concurrence, and quantum geometric invariance.
Antilinear superoperator $\mathcal{M}$ is called involution if $\mathcal{M}^2=\mathcal{I}$, or skew involution if $\mathcal{M}^2=-\mathcal{I}$.

\begin{definition}
	An antilinear superoperator $\Theta$ is called a generalized $\Theta$-conjugation if $\Theta$ is (i) weak unitary $\Theta^{\ddag}=\Theta^{-1}$,
	(ii) Hermitian $\Theta^{\ddag}=\Theta$,
	(iii) involution $\Theta^{-1}=\Theta$,
	and (iv) unital $\Theta(\one)=\one$.
	Similarly, $\Theta$ is called a generalized skew $\Theta$-conjugation if it is (i) weak unitary $\Theta^{\ddag}=\Theta^{-1}$,
	(ii) skew Hermitian $\Theta^{\ddag}=-\Theta$,
	(iii) skew involution $\Theta^{-1}=-\Theta$,
	and (iv) unital $\Theta(\one)=\one$.
\end{definition}

The relations between different classes of antilinear superoperators are shown in Fig.~\ref{fig:relation}.

The $\Theta$-conjugation \cite{Uhlmann2000fidelity}, by definition, is a (strong) unitary involution, making it a special case of generalized $\Theta$-conjugation. Typical examples include the time-reversal operation \cite{wigner2012group}, and Hill-Wootters conjugation \cite{Hill1997entanglement}.
Using the polar decomposition of antilinear involution, we will obtain several crucial criteria for determining whether a given involution is a  generalized $\Theta$-conjugation or not. The detailed discussion is given in Appendix \ref{app:AntiSup}.

The $\Theta$-conjugated quantum fidelity and quantum concurrence play key roles in studying quantum information and quantum correlations. Here, we can also introduce their counterparts for generalized $\Theta$-conjugation. 

Fidelity measures how close two states are. For completeness, we can introduce the $p$-norm fidelity between quantum states $\rho$ and $\sigma$,
\begin{equation}
	F_p(\rho,\sigma)=\| \sqrt{\rho} \sqrt{\sigma}\|_p.
\end{equation}
Here the $p$-norm of $A=\sqrt{\rho}\sqrt{\sigma}$ is defined by $\|A\|_p=(\Tr((A^\dagger A)^{p/2}))^{1/p}=(\Tr(\sqrt{\sqrt{\sigma}\rho\sqrt{\sigma}})^{p})^{1/p}$. When taking $p=1$, we obtain the normal fidelity $F(\rho,\sigma)=\Tr (\sqrt{\sqrt{\sigma} \rho \sqrt{\sigma} })$. 
Many properties of fidelity can be generalized to $p$-norm fidelity; see Appendix \ref{App:FidelityP} for details.

Following the definition of $\Theta$-fidelity \cite{Uhlmann2000fidelity}, we introduce the definition below:

\begin{definition}
	For a given generalized $\Theta$-conjugation, the $p$-norm generalized $\Theta$-fidelity between two states is defined as
	\begin{equation}
		F_{\Theta,p}(\rho,\sigma):=F_p(\rho, \Theta(\sigma))=\| \sqrt{\rho} \sqrt{\Theta(\sigma)}\|_p.
	\end{equation}
	By abbreviating $\tilde{\rho}:=\Theta (\rho)$, the $p$-norm generalized $\Theta$-fidelity of $\rho$ is defined as
	\begin{equation}
		F_{\Theta,p}(\rho):=F_{p}(\rho, \tilde{\rho}).
	\end{equation}
	The most commonly used one is the $1$-norm case.
\end{definition}
Notice that generalized $\Theta$-conjugation may not be an antilinear channel; therefore, it may map a positive semidefinite operator to an operator that is not positive semidefinite.  To remedy this problem, we can shrink the Bloch vector corresponding to $\Theta(\rho)$ such that it becomes a vector in the Bloch body (see Sec.~\ref{sec:geo} for a detailed discussion). 


Similar to fidelity, we can also introduce the generalized $\Theta$-concurrence. Consider the operator $\sqrt{\rho}\sqrt{\sigma}$, whose singular values are nonnegative numbers: $\lambda_1\geq \lambda_2\geq \cdots \geq \lambda_d$. Consider the $p$-th power of these singular values, {\it viz.}, the spectrum of $(\sqrt{\rho} \sigma \sqrt{\rho} )^{p/2}$; the concurrence of order $p$ between $\rho$ and $\sigma$ is defined as
\begin{equation}
	C_p(\rho,\sigma)=\max \left\{0,\left(\lambda_1^p-\sum_{j=2}^d \lambda_j^p\right)^{1/p}\right\}.
\end{equation}
When taking $p=1$, we obtain the normal definition of concurrence between $\rho$ and $\sigma$, $	C(\rho,\sigma)=\max \{0,\lambda_1-\sum_{j=2}^d \lambda_j\}$.
Following the definition of $\Theta$-concurrence \cite{Uhlmann2000fidelity}, we introduce the definition below:
\begin{definition}
	For a fixed generalized $\Theta$-conjugation, consider the spectrum of the operator 
	$(\sqrt{\rho} \Theta(\sigma) \sqrt{\rho} )^{p/2}$. By ordering its eigenvalues in a descending order $\lambda_1^p \geq \lambda_2^p\geq \cdots \geq \lambda_d^p$, 
	the generalized $\Theta$-concurrence of order $p$ between $\rho$ and $\sigma$ is defined as
	\begin{equation}
		C_{\Theta,p}(\rho,\sigma)=\max \left\{0,\left(\lambda_1^p-\sum_{j=2}^d \lambda_j^p\right)^{1/p}\right\}.
	\end{equation}
	By abbreviating $\tilde{\rho}:=\Theta (\rho)$, the generalized $\Theta$-concurrence of order $p$ for $\rho$ is defined as $C_{\Theta,p} (\rho)=C_{\Theta,p} (\rho,\tilde{\rho})$.
\end{definition}
As we will see later, for qubit systems, quantum concurrence is related to the Bloch vectors in a simple way.
Similar to generalized $\Theta$-conjugated fidelity, there is also a problem with positivity, which will be discussed in the next section.

From theorem~\ref{thm:TPC} in Sec.~\ref{sec:unital}, the generalized $\Theta$-conjugation is an antilinear TP map. The generalized $\Theta$- fidelity and concurrence satisfy the following properties:
\begin{enumerate}
    \item Suppose we have the Kraus representation for generalized $\Theta$-conjugation $\Theta(\rho)=\sum_j A_j \rho^* B_j$. For unitary $U$, define $\Theta'$ as 
$\Theta'(\rho)=\sum_j A'_j \rho^* {B_j'}^{\dagger}$ with $A'_j=U^{\dagger}A_j U$ and ${B_j'}^{\dagger}=UB_j^{\dagger} U^{\dagger}$. Then we have $F_{\Theta,p}(\rho)=F_{\Theta',p}(\rho)$ and $C_{\Theta,p}(\rho)=C_{\Theta',p}(\rho)$.

\item Both generalized $\Theta$-fidelity and $\Theta$-concurrence are bounded $0\leq F_{\Theta,p}, C_{\Theta,p}\leq 1$ after shrinking the Bloch vectors for generalized $\Theta$-conjugation to ensure they are positive.

\item They are homogeneous. Thus for non-negative real number $\lambda$, we have $F_{\Theta,p}(\lambda\rho)=\lambda F_{\Theta,p}(\rho)$ and $C_{\Theta,p}(\lambda \rho)=C_{\Theta,p}(\lambda \rho)$.

\item From the concavity of 1-fidelity and antilinearity of $\Theta$, we obtain $F_{\Theta,1}(\sum_i p_i \rho_i)\geq \sum_i p_i F_{\Theta,1}(\rho_i)$.

\end{enumerate}

Notice that when the Kraus rank of a generalized $\Theta$-conjugation is one, the generalized $\Theta$-conjugation reduces to a $\Theta$-conjugation. In this case, Uhlmann proved that $\Theta$-concurrence and $\Theta$-fidelity are related to each other~\cite{Uhlmann2000fidelity}:
$\Theta$-concurrence is the largest convex function on the state space that coincides with $|\langle \psi|\Theta |\psi\rangle|$ on pure states, while $\Theta$-fidelity is the smallest concave function on the state space that coincides with $|\langle \psi|\Theta |\psi\rangle|$ on pure states.
For two-qubit cases, $\Theta$-concurrence being equal to zero is a necessary and sufficient condition for the state to be separable, while $\Theta$-fidelity being equal to one is a necessary and sufficient condition for the state to be in the convex hull of the maximally entangled pure states.
Similar to the qubit case, we propose generalized $\Theta$-concurrence and generalized $\Theta$-fidelity based on generalized $\Theta$-conjugation. Since generalized $\Theta$-conjugation applies to arbitrary local quantum dimensions, not just qubits, we conjecture that these two concepts play the same role as their qubit counterparts.
This will be left for future study.

\section{Application: quantum geometric invariance for qudit system}
\label{sec:geo}

With the above preparation, we are now in a position to discuss the quantum geometric invariance of qudit systems.
Consider the linear isomorphism $\Phi: \mathbf{H}(\mathbb{C}^d) \to \mathbb{R}^{d^2}$
given by Bloch representation (details provided below). By mapping $\sigma_{\mu}$ to $e_{\mu}=(0,\cdots,0,1,0,\cdots,0)$, each Hermitian matrix $\rho$ has a corresponding vector $x(\rho)=(x_0,\cdots,x_{d^2-1})$. 
Conversely, for every $x\in  \mathbb{R}^{d^2}$, we obtain a Hermitian matrix of the form $\rho(x)=\frac{1}{d}\sum_{\mu=0}^{d^2-1} x_{\mu} \sigma_{\mu}$.
We aim to study transformations of states corresponding to the geometric transformations of Bloch space-time vector $x_{\mu}$.

For a given  superoperator $\mathcal{E}$, which maps $\rho$ to $\mathcal{E}(\rho)=\frac{1}{d}(\sum_{\mu}x'_{\mu}\sigma_{\mu})$, there is a corresponding geometric transformation $T_{\mathcal{E}}: x_{\mu}\mapsto x'_{\mu}(x)$. For a given geometric transformation $T$, there also exists a corresponding superoperator $\mathcal{E}_T$. See Fig.~\ref{fig:bloch} for a depiction.
For these geometric transformations, like Lorentz transformation, investigating the corresponding transformations of states plays a crucial role in studying entanglement, monogamy relations, Jones vector transformation in quantum optics, and so on \cite{Han1997stokes,Han1999wigner,Verstraete2002lorentz,Jaeger2003quantum,Eltschka2015monogamy}. 

We regard $x\in \mathbb{R}^{d^2}$ as a space-time vector. By introducing different generalized $\Theta$-conjugations, different geometric structures on  $\mathbb{R}^{d^2}$ are obtained. In this different geometric space-time, different geometric transformations and their properties will be discussed.

\subsection{Bloch representation of qudit}
\label{sec:bloch}

To start with, let us first recapitulate the definition of Bloch representations of qudit states and their properties. \cite{Hioe1981N,Jakobczyk2001geometry,Kimura2003bloch,eltschka2021shape}. 
Consider a $d$-dimensional system $\mathcal{H}=\mathbb{C}^{d}$ with standard basis $\{|k\rangle|k=0,1\cdots, d-1\}$, associated operator space $\mathbf{B}(\mathcal{H})$ equipped with Hilbert-Schmidt inner product $\langle A,B\rangle_{\rm HS}=\operatorname{Tr} (A^{\dagger} B)$ is a Banach space.
The set of all Hermitian matrices $\mathbf{H}(\mathcal{H})$ forms a $d^2$-dimensional real linear subspace of  $\mathbf{B}(\mathcal{H})$. 
The set of all density operators $\mathbf{D}(\mathcal{H})$ is a convex subset of $\mathbf{H}(\mathcal{H})$ consisting of all positive semidefinite trace-one operators:
\begin{equation}
	\mathbf{D}(\mathcal{H})=\{\rho \in \mathbf{H}(\mathcal{H})\,|\, \rho \geq 0; \Tr(\rho)=1 \}.
\end{equation}
The Bloch representation is a linear map from $\mathbf{D}(\mathcal{H})$ to $\mathbb{R}^{d^2}$.

The Bloch representation is in general not unique.
To construct a Bloch representation, first we need to choose a basis of $ \mathbf{H}(\mathcal{H})$.
It is convenient to use Hilbert-Schmidt basis $\{\sigma_{\mu}\,|\,\mu=0,\cdots,d^2-1\}$ which satisfies
\begin{itemize}
	\item  The basis includes the identity operator $\sigma_0=\mathds{I}$; 
	\item  $\Tr( \sigma_{j})=0$ for all $j \geq 1$; 
	\item These matrices are orthogonal $	\Tr(\sigma_{\mu}\sigma_{\nu})=d\delta{\mu \nu}$.
\end{itemize}
For the $d=2$ case, Pauli matrices form such a basis.
For the $d=2^n$ case, tensor products of Pauli matrices form such a basis.
For the more general cases, a typical explicit matrix representation of such a basis is generalized Gell-Mann (GGM) matrices \cite{Gell-mann1962symmetry} which consists of 
\begin{itemize}
	\item[(1)] $\frac{d(d-1)}{2}$ symmetric GGMs which correspond to Pauli X-matrix
	\begin{equation}
		\Lambda_{s}^{j k}=\sqrt{\frac{d}{2}}\left(|j\rangle\langle k|+| k\rangle\langle j|\right), \quad 0 \leq j<k \leq d-1;
	\end{equation}
	\item[(2)] $\frac{d(d-1)}{2}$ antisymmetric GGMs which correspond to Pauli Y-matrix
	\begin{equation}
		\Lambda_{a}^{j k}=\sqrt{\frac{d}{2}}\left(-i|j\rangle\langle k|+i| k\rangle\langle j|\right), \quad 0 \leq j<k \leq d-1;
	\end{equation}
	\item[(3)]  $(d-1)$ diagonal GGMs which correspond to Pauli Z-matrix
	\begin{equation}
	\Lambda^{l}=\sqrt{\frac{d }{(l+1)(l+2)}}\left(\sum_{j=0}^{l}|j\rangle\langle j|-(l+1)| l+1\rangle\langle l+1|\right),\quad 0 \leq l \leq d-2;
	\end{equation}
	\item[(4)]  The identity matrix $\mathds{I}$.
\end{itemize}
There are in total $\frac{d(d-1)}{2}+\frac{d(d-1)}{2}+(d-1)+1=d^2$ matrices. Notice that GGM matrices are generators of $\mathfrak{su}(d)$ and also serve as a basis of complex vector space $\mathbf{B}(\mathcal{H})$.

Since $\mathbf{H}(\mathcal{H}) \simeq \mathbb{R}^{d^2}$, a density operator $\rho$ can be uniquely represented in Hilbert-Schmidt basis as 
\begin{equation}\label{eq:Bloch}
	\rho=\frac{1}{d}\sum_{\mu=0}^{d^2-1}x_{\mu}\sigma_{\mu}=\frac{1}{d}(\sigma_0+\vec{x}\cdot \vec{\sigma}),
\end{equation}
where $x_{\mu}\in \mathbb{R}$ and $x_0=1$ since all $\sigma_{\mu}$ are traceless except $\sigma_0$ and the density operator is trace-one.  The $(d^2-1)$-dimensional vector $\vec{x}$ is called a Bloch vector (or coherence vector). 
Notice that all GGM matrices are Hermitian, thus, they can be regarded as observables. One of the advantages for this kind of representation is that 
\begin{equation}
	\langle \sigma_{\mu}\rangle= \Tr (\sigma_{\mu}\rho)=x_{\mu}.
\end{equation}
By measuring the expectation value of these $d^2-1$ observables, we can determine the state \cite{Hioe1981N}.

From the condition that purity $\operatorname{Tr} (\rho^2)\leq 1$, we see that
\begin{equation}\label{eq:purity}
	\| \vec{x}\|^2\leq d-1.
\end{equation}
For pure states, $\| \vec{x}\|^2= d-1$, and for mixed states, $\| \vec{x}\|^2< d-1$.
Notice that Eq.~(\ref{eq:purity}) is not sufficient for $\rho$ in Eq.~(\ref{eq:Bloch}) to be a density operator, so the condition that $\rho \geq 0$ still needs to be imposed \cite{Hioe1981N,Jakobczyk2001geometry,Kimura2003bloch}. 
It has been shown that the angle between any two pure-state Bloch vectors $\vec{x}$ and $\vec{y}$ must satisfy \cite{Jakobczyk2001geometry}
\begin{equation}
	-\frac{1}{d-1} \leq \cos (\theta_{\vec{x},\vec{y}}) \leq 1.
\end{equation}
This implies that the set of all Bloch vectors for qudit states forms a convex subset of $(d^2-1)$-dimensional ball $B^{d^2-1}(0;\sqrt{d-1})$ with radius $\sqrt{d-1}$.
To impose the condition that $\rho\geq 0$ (namely, all eigenvalues are nonnegative), we consider the characteristic polynomial 
\begin{equation}\label{eq:poly}
	\det (\lambda \one -\rho)=\sum_{j=0}^d (-1)^j a_j \lambda^{d-j}.
\end{equation}
Using Vieta's theorem 
$a_j=\sum_{1\leq k_1<\cdots <k_j\leq d} \lambda_{k_1}\cdots \lambda_{k_j}$, it can be proved that $\rho \geq 0$ is equivalent to $a_0,\cdots,a_d \geq 0$ \cite{Kimura2003bloch}.
With this, each Bloch vector $\vec{x}$ corresponds to a set of real coefficients of the characteristic polynomial, $a_j(\vec{x})$. Thus the Bloch convex body corresponding to the set of all density matrices can be defined as 
\begin{equation}
	\mathcal{B}(d^2-1)=\{\vec{x} \in \mathbb{R}^{d^2-1}\,|\, a_j(\vec{x}) \geq 0, \forall j\}.
\end{equation}
For $d=2$, the Bloch body is exactly a ball. However the shapes are very complicated for higher-dimensional cases.

\begin{example}[3-dimensional Bloch convex body]\label{exp:3d}
	For 3-dimensional case, the 9 GGM matrices are: 
	\begin{itemize}
		\item[(1)] 3 symmetric matrices
		\begin{align}
			\sigma_1=\sqrt{\frac{3}{2}} \left(\begin{array}{lll}
				0 & 1 & 0 \\
				1 & 0 & 0 \\
				0 & 0 & 0
			\end{array}\right),  \quad
			\sigma_4=\sqrt{\frac{3}{2}} \left(\begin{array}{lll}
				0 & 0 & 1 \\
				0 & 0 & 0 \\
				1 & 0 & 0
			\end{array}\right), \quad 
			\sigma_6=\sqrt{\frac{3}{2}} \left(\begin{array}{lll}
				0 & 0 & 0 \\
				0 & 0 & 1 \\
				0 & 1 & 0
			\end{array}\right).
		\end{align}
		\item[(2)] 3 antisymmetric matrices
		\begin{align}
			\sigma_2=	\sqrt{\frac{3}{2}} \left(\begin{array}{ccc}
				0 & -i & 0 \\
				i & 0 & 0 \\
				0 & 0 & 0
			\end{array}\right), \quad
			\sigma_5=\sqrt{\frac{3}{2}} \left(\begin{array}{ccc}
				0 & 0 & -i \\
				0 & 0 & 0 \\
				i & 0 & 0
			\end{array}\right), \quad
			\sigma_7=\sqrt{\frac{3}{2}} \left(\begin{array}{ccc}
				0 & 0 & 0 \\
				0 & 0 & -i \\
				0 & i & 0
			\end{array}\right).
		\end{align}
		\item[(3)] 2 diagonal matrices
		\begin{align}
			\sigma_3=\sqrt{\frac{3}{2}} \left(\begin{array}{ccc}
				1 & 0 & 0 \\
				0 & -1 & 0 \\
				0 & 0 & 0
			\end{array}\right), \quad
			\sigma_8=\frac{1}{\sqrt{2}}\left(\begin{array}{ccc}
				1 & 0 & 0 \\
				0 & 1 & 0 \\
				0 & 0 & -2
			\end{array}\right).
		\end{align}
		
		\item[(4)] 1 identity operator $\sigma_0=\mathds{I}$.	
	\end{itemize}
	
	These are $\mathfrak{su}(3)$ generators and satisfy the commutation relation
	\begin{equation}
		[\sigma_j, \sigma_{k}]= i 2 \sqrt{\frac{3}{2} }  \sum_l f_{jkl} \sigma_l,
	\end{equation}	
	where the nonzero structure constants are: $f_{123}=1$, $f_{147}=f_{246}=f_{257}=f_{345}=1/2$, $f_{156}=f_{367}=-1/2$, $f_{246}=1/2$ and $f_{458}=f_{678}=\sqrt{3}/2$.
	The anticommutation relation is
	\begin{equation}
		\{\sigma_j, \sigma_{k}\}=  2 \delta_{jk} \one +\sqrt{6} \sum_l g_{jkl} \sigma_l,
	\end{equation}	
	where nonzero structure constants are $g_{118}=g_{228}=g_{338}=-g_{888}=1/\sqrt{3}$, $g_{448}=g_{558}=g_{668}=g_{778}=-1/2\sqrt{3}$, $g_{146}=g_{157}=g_{256}=g_{344}=g_{355}=-g_{247}=-g_{366}=-g_{377}=1/2$.

	For a density operator $\rho=\frac{1}{3} (\one +\sum_{i=1}^8x_i\sigma_i)$, the coefficient of characteristic polynomial in Eq.~(\ref{eq:poly}) can be explicitly calculated using Newton identities, 
	\begin{equation}
		k a_k=\sum_{j=1}^k (-1)^{j-1} N_j a_{k-j}, k=1,2,3
	\end{equation}
	where $N_j=\sum_{l=1}^3 (\lambda_l)^j$ is the $j$-th power sum of all eigenvalues of $\rho$. In this way,
	we see that 
	\begin{equation}
		\begin{aligned}
			&	1! a_1=N_1=\Tr \rho=1,\\
			&	2! a_2=N_1^2-N_2=1-\Tr (\rho^2), \\
			&	3! a_3= N_1^2-N_1 N_2-2 N_2 +2 N_3=1- 3 \Tr (\rho^2) +2 \Tr (\rho^3).
		\end{aligned}
	\end{equation}
	The condition that $\rho \geq 0$ now becomes $j!a_j\geq 0$ for all $j$, which can be explicitly calculated by using structure constants $f_{jkl}$ and $g_{jkl}$ of commutation and anti-commutation relations $\sigma_{j}\sigma_k= \frac{1}{2}([\sigma_j,\sigma_k]+\{\sigma_j,\sigma_k\})=\frac{1}{2}( i 2 \sqrt{\frac{3}{2} }  \sum_l f_{jkl} \sigma_l+2 \delta_{jk} \one +\sqrt{6} \sum_l g_{jkl} \sigma_l)$. See Ref. \cite{Jakobczyk2001geometry} for more details.
\end{example}

\subsection{Generalized $\Theta$-conjugation and space-time metric}

One of the well-known examples of $\Theta$-conjugation is the
Hill-Wootters spin-flip operation $\rho\mapsto \sigma_y \mathcal{K}(\rho) \sigma_y$ \cite{Hill1997entanglement,Wootters1998entanglement}.
The corresponding geometric transformation is the parity transformation: $x_0\mapsto x_0, \vec{x}\mapsto -\vec{x}$.
This can be naturally generalized to the qudit case. Consider a special generalized $\Theta$-conjugation, which, when acting on GGM matrices, has the form
\begin{equation}
	\Theta(\sigma_0)=\sigma_0, \Theta(\sigma_j)=-\sigma_j, 1\leq j \leq d^2-1.
\end{equation}
 Namely, $\Theta$ acting on the Bloch representation of a Hermitian matrix $\rho=\sum_jx_j\sigma_j$ is of the form 
\begin{equation}
    \Theta(\rho) = x_0\sigma_0 - \sum_{j=1}^{d^2-1}x_\mu\sigma_j = x_0\sigma_0 - \vec{x}\cdot \vec{\sigma}. 
\end{equation} 
We then use the antilinear extension $\Theta(\lambda_1\rho_1+\lambda_2\rho_2):=\bar{\lambda}_1\Theta(\rho_1)+\bar{\lambda}_2\Theta(\rho_2)$ to give an antilinear superoperator.  
It is easily checked that $\Theta^{\ddag} =\Theta^{-1}$.
Notice that under this generalized $\Theta$-conjugation, a density operator $\rho$ may be mapped to a Hermitian operator $\Theta (\rho)$ with negative eigenvalues. This is because that the Bloch convex body for the $d\geq 3$ case does not have rotational symmetry.
To remedy this problem, we can shrink the Bloch vector $\vec{x}'=f_{\Theta} (\vec{x})$ such that $\vec{x}''=\vec{x}'/\lambda$ gives a density operator. 
The shrinking process works as follows: Suppose that $\vec{a}_{\vec{x}'}$ is the Bloch vector in $\vec{x}'$ direction with the maximum length $a$, then $\vec{x}'':=a \vec{x}'/\sqrt{d-1}$. This shrinking process may break the linearity in general. So for a given convex combination $p \rho_1 +(1-p)\rho_2$, we must first calculate the corresponding overall Bloch vector, then map it to a new Bloch vector. This can remedy the problems when we define generalized $\Theta$-conjugated fidelity and concurrence.

We can similarly consider the generalized $\Theta$-conjugation corresponding to the partial parity transformation, 
\begin{equation}
	\Theta(\sigma_{j_k})=-\sigma_{j_k}, \forall\, k=1,\cdots, q,
\end{equation}
while leaving all other $\sigma_j$ unchanged. In terms of the Bloch representation of a Hermitian matrix $\rho = \sum_jx_j\sigma_j$, this is written explicitly as 
\begin{equation}
    \Theta(\rho) = x_0\sigma_0 - \sum_{j=1}^{q}x_j\sigma_j + \sum_{j=q+1}^{d^2-1}x_j\sigma_j=x_0\sigma_0 - \vec{x}_a\vec{\sigma}_a+\vec{x}_b\vec{\sigma}_b,
\end{equation}
where $\vec{x}_a = (x_1,\cdots,x_{q})$, $\vec{x}_b = (x_{q+1},\cdots,x_{d^2-1})$ (similar for $\vec{\sigma}_a$ and $\vec{\sigma}_b$). The antilinear extension is also utilized to make $\Theta$ an antilinear superoperator.
The corresponding Bloch vector transformation is \begin{equation}
	\begin{aligned}
		x_{j_k}\mapsto &-x_{j_k}, \forall k=1,\cdots, q, \\
		x_{j_k}\mapsto x_{j_k}, &\forall k=q+1,\cdots, d^2-1.
	\end{aligned}
\end{equation}
This kind of generalized $\Theta$-conjugation is crucial for us to study Lorentzian invariance of qudit state.

\subsection{Quantum Euclidean invariance for qudit system}

\begin{figure}[t]
	\centering
	\includegraphics[width=7cm]{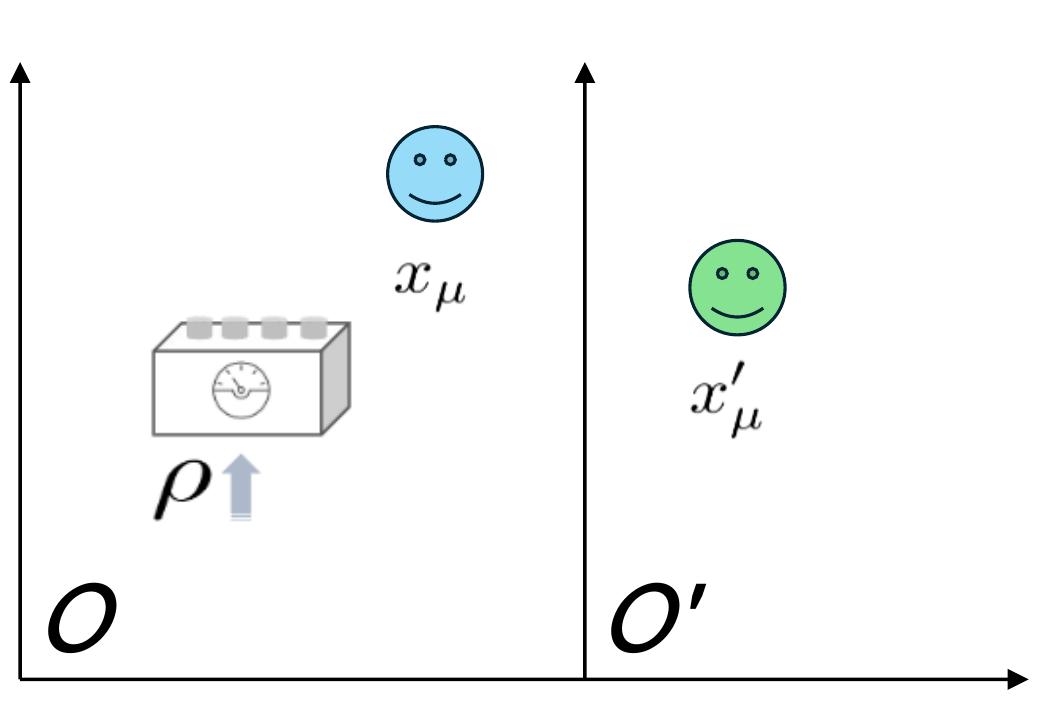}\\
	\caption{An illustration of the geometric transformation of quantum states for Bloch space-time vector.}\label{fig:bloch}
\end{figure}

Let us now consider Euclidean invariance of single qudit state. We assume that the space $\mathbb{R}^{d^2}$ is equipped with the Euclidean  norm $\|x\|_{E}=\sum_{\mu \nu}\delta_{\mu\nu}x_{\mu}x_{\nu}=\sum_{\mu=0}^{d^2-1} x_{\mu}^2$. 
In $\mathbf{H}(\mathbb{C}^{d})$, the norm corresponds to
\begin{equation}
	\langle \rho,\rho\rangle_{\rm HS}=	\Tr (\rho^2)=\frac{1}{d} \|x(\rho)\|_{E}^2=\frac{1}{d}(1+\|\vec{x}\|^2),
\end{equation}
which is nothing but the purity of the state. Thus the Euclidean invariance for single qudit state is quantum transformations which preserve the purity of states. In Bloch representation, this corresponds to the orthogonal group $O(d^2-1,\mathbb{R})$.
Since the qudit Bloch representation does not have rotational symmetry, we still need to do some shrinking if the rotated vector is not a Bloch vector that corresponds to a positive semidefinite operator.

For the bipartite case, it is convenient to introduce the joint observables $x_{\mu\nu}=\Tr(\rho \sigma_{\mu}\otimes \sigma_{\nu})$ and express the state as
\begin{equation}
\begin{aligned}
	\rho&=\frac{1}{d^2}\sum_{\mu,\nu} x_{\mu\nu} \sigma_{\mu}\otimes \sigma_{\nu} \\
	&=\frac{1}{d^2} \left(\sigma_0\otimes \sigma_0+\sum_{j\geq1} x_{j0}  \sigma_{j}\otimes \sigma_0 +  \sum_{j\geq1}  x_{0j}\sigma_0\otimes \sigma_j+\sum_{j,k \geq 1} x_{jk}  \sigma_j\otimes \sigma_k \right). \label{eq:biqudit}
\end{aligned}
\end{equation}
The tensor product of Hilbert-Schmidt matrices $\sigma_{\mu}\otimes \sigma_{\nu}$ provides a basis for the $\mathbf{H}(\mathbb{C}^d)\otimes \mathbf{H}(\mathbb{C}^d)$.
In this case, the Euclidean norm is
\begin{equation}
	\langle \rho ,\rho\rangle_{\rm HS}=\Tr (\rho^2) =	\frac{1}{d^2}\sum_{\mu \nu} x_{\mu \nu}^2.
\end{equation}
Since $x_{00}\equiv 1$, we see that the operation preserving the Euclidean norm is the group $O(d^4-1,\mathbb{R})$. 
For the $n$-qudit case, the generalization is straightforward,
\begin{equation}
	\rho=\frac{1}{d^n}\sum_{\mu_1\cdots \mu_n} x_{\mu_1\cdots\mu_n}  \sigma_{\mu_1}\otimes \cdots \otimes \sigma_{\mu_n}.
\end{equation}
The Euclidean norm is 
\begin{equation}
	\langle \rho ,\rho\rangle_{\rm HS}=\Tr (\rho^2) =	\frac{1}{d^n}\sum_{\mu_1\cdots \mu_n} x_{\mu_1\cdots \mu_n}^2.
\end{equation}
The symmetry group is $O(d^{2n}-1,\mathbb{R})$.
Thus we see that the Euclidean norm corresponds to the purity of the state, and it can be expressed as the Hilbert-Schmidt inner product of $\rho$ with itself. The Euclidean invariance of quantum states is the equivalent class of states which is invariant under the group $O(d^{2n}-1,\mathbb{R})$.
To consider a more general case, i.e., the Euclidean rotation over the $d^{2n}$ space, we need to do some shrinking of the Bloch tensors we obtained first, then the Euclidean-invariant class of state will have more general meaning and the symmetry group will be $O(d^{2n},\mathbb{R} )$.

\subsection{Quantum Lorentzian invariance for qudit system}

To investigate the Lorentzian invariance, we need to consider the space-time $\mathbb{R}^{p,q}$ with the Minkowskian metric defined as $\eta^{\mu\nu}=\mathrm{diag}(+,\cdots,+,-,\cdots,-)$, where $p+q=d^2$. 
The Lorentzian norm is thus $\|x\|_L^2=\sum_{\mu,\nu}\eta^{\mu\nu}x_{\mu}x_{\nu}=\sum_{\mu=0}^{p-1}x_{\mu}^2-\sum_{\mu=p}^{p+q-1}x_{\mu}^2$.
The Lorentz transformation is a linear transformation
\begin{equation}
	x'_{\mu}=\sum_{\nu}\Lambda_{\mu}^{\,\,\nu} x_{\nu}
\end{equation}
such that $\|x'\|^2=\|x\|^2$. The set of all Lorentz transformations forms Lorentz group $O(p,q)=\{\Lambda\in GL(p+q;\mathbb{R})\,|\, \Lambda^T \eta \Lambda =\eta \}$.

For single-qubit state in Bloch representation $\rho=\frac{1}{2}(\one+\vec{x}\cdot \vec{\sigma})$, the Hill-Wootters conjugation gives $\tilde{\rho}=\sigma_y \rho^* \sigma_y =\frac{1}{2}(\one-\vec{x}\cdot \vec{\sigma})$, then we see that
\begin{align}
	4 \det \rho = 2 \Tr (\rho \tilde{\rho})
	=x_0^2-x_1^2-x_2^2-x_3^2
	=\sum_{\mu,\nu}\eta^{\mu\nu}x_{\mu}x_{\nu}.
\end{align}
For higher-dimensional situations, we can utilize the generalized $\Theta$-conjugation which maps $\sigma_j$ to $-\sigma_j$ for $j\geq p$ and leaves all other Hilbert-Schmidt basis matrices unchanged. We introduce the generalized $\Theta$-conjugated Hilbert-Schmidt inner product $\langle \rho, \chi\rangle_{\Theta}:=	\langle \rho,\Theta(\chi)\rangle_{\rm HS}=\Tr(\rho \Theta(\chi))$. Then for Hilbert-Schmidt operator, we have $\langle \sigma_{\mu},\sigma_{\nu}\rangle_{\Theta}=\eta_{\mu\nu}$, hence the Lorentzian norm for a density operator can be expressed as
\begin{equation}
	\langle \rho,\rho\rangle_{\Theta}=\Tr(\rho \Theta(\rho))=\frac{1}{d}\sum_{\mu,\nu}\eta_{\mu\nu}x_{\mu}x_{\nu}.
\end{equation}
We usually take $p=1$, namely, $\Theta(\sigma_0)=\sigma_0$ and $\Theta(\sigma_j)=-\sigma_j$ for all $j\geq 1$.

For bipartite state $\rho$ as in Eq.~(\ref{eq:biqudit}), we introduce $R_{\rho}=\rho \Theta^{\otimes  2} (\rho)$, then the Lorentzian norm is therefore
\begin{equation}
    \begin{aligned}
    	\langle \rho ,\rho\rangle_{\Theta}&=\Tr (R_{\rho} )= x^2_{00}-\sum_{j=1}^{d^2-1} (x^2_{0j}  +x_{j0}^2)+\sum_{j,k=1}^{d^2-1}x^2_{jk}\\
    	&= L_0-L_1+L_2.
    \end{aligned}
\end{equation}
Hereinafter, $L_k$ denotes the sum of all terms $x^2_{j_1\cdots j_n}$ with $k$ space-like indices; it is called sector length \cite{aschauer2003local,wyderka2020characterizing}, which is invariant under local unitary transformations and can be used to describe the correlation structure of quantum states. 

Similar to Euclidean case, the Lorentzian invariance is characterized by equivalence class which is invariant under Lorentz transformation. One of the main problems here is that the Lorentz boost may map a Bloch space-time vector $x_{\mu}$ to the one with time component $x'_0\neq 1$. This can also be remedied by shrinking or dilating the Bloch space-time vectors.

\subsection{Quantum geometric invariance and entanglement distribution}
\label{sec:mono}

One of the characteristic features of quantum correlations is that they cannot be shared freely in a many-body system. This phenomenon is now known as monogamy of quantum correlations. It is shown that there exist monogamy relations for Bell nonlocality  \cite{Pawlowski2009monogamy,Kurzynski2011correlation,Jia2016,jia2017exclusivity,Jia2018entropic}, quantum steering \cite{Reid2013monogamy}, and entanglement \cite{CKW2000,Koashi2004monogamy,Eltschka2015monogamy,eltschka2018distribution,eltschka2020maximum}.
In this subsection, we study monogamy equalities of entanglement restricted by quantum geometric invariance.

Now consider von Neumann entropy $S(\rho) = -\Tr (\rho \ln \rho)$. Using Mercator series for $\ln \rho= \ln (\one + (\rho -\one)) \simeq  (\rho -\one) -(\rho-\one)^2/2+\cdots$ and keeping only linear term, we obtain a quantity called linear entropy
\begin{equation}
	S_L(\rho) = \Tr (\rho) -\Tr (\rho^2).
\end{equation}
Note that many authors made the convention for linear entropy differing from the one we give here with a multiple two.
For two-qubit state $\psi_{AB}$, linear entropy of the reduced state are related to concurrence of the state via
\begin{equation}
	2S_L(\rho_A)=	2S_L(\rho_B)=C_{A:B}^2(\psi_{AB}),
\end{equation}
where $\rho_A = \Tr_B(\rho_\psi)$, $\rho_B = \Tr_A(\rho_\psi)$ with $\rho_\psi=|\psi_{AB}\rangle\langle\psi_{AB}|$. From this, we see that linear entropy can be used to measure quantum correlation of the state.  
Using this observation, Eltschka and Siewert derive the distribution of quantum correlations from quantum Lorentzian invariance of qubit states \cite{Eltschka2015monogamy}. Here we generalize this result to qudit case (see also \cite{eltschka2018distribution}).

For $n$-qudit state $\rho$, the Euclidean norm is 
\begin{equation}
	d^n \Tr ( \rho^2 )= L_0+L_1+\cdots +L_n.
\end{equation}
Consider the generalized $\Theta$-conjugation $\Theta(\sigma_0)=\sigma_0$ and $\Theta(\sigma_j)=-\sigma_j$ ($j\geq 1$). If we set $R_\rho=\rho \Theta^{\otimes n}(\rho)$, the Lorentzian norm is given by
\begin{equation}
	d^n\Tr (R_{\rho}) =L_0-L_1+\cdots +(-1)^n L_n.
\end{equation}
It is easy to see that
\begin{equation}\label{eq:R}
	(-1)^n d^n\Tr (R_{\rho}) =d^n \Tr (\rho^2) -2 \sum_{k=\delta_n}^{\lfloor n/2 \rfloor}L_{2k -\delta_n},
\end{equation}
where $\delta_n=(1+(-1)^n)/2$.
Let us denote the bipartition of $n$-particle system as $\mathcal{A}|\mathcal{A}^c$. The linear entropy $S_L({\mathcal{A}}):=S_L(\rho_\mathcal{A})$ for this bipartition measures the entanglement between $\mathcal{A}$ and $\mathcal{A}^c$.
It is straightforward to check that
\begin{equation}
	\Tr (R_{\rho})  =\sum_{\mathcal{A}|\mathcal{A}^c} a_{\mathcal{A}|\mathcal{A}^c} S_L({\mathcal{A}}), \label{eq:dis}
\end{equation}
where we have denoted the linear entropy for empty partition as $S_{L} (\emptyset) =1$.
This characterizes the distribution of the entanglement for a multipartite state, since $ \Tr (R_{\rho})$ measures the overall entanglement of the state $\rho$.

\begin{example}
	Consider an $n$-qubit state $\psi$, the coefficient in Eq.~(\ref{eq:dis}) has been carefully calculated in Ref.~\cite{Eltschka2015monogamy}:
	\begin{equation}
		\Tr (R_{\psi}) = (-1)^{n+1} S_{L}(\rho)+\sum_{\mathcal{A},\mathcal{A}^c\neq \emptyset} (-1)^{|\mathcal{A}^c|+n+1} S_L(\mathcal{A}) .
	\end{equation}
	This characterizes how entanglement is distributed over these $n$ particles. For $n=3$ case, the constraint becomes trivial $S_L(\mathcal{A})=S_L(\mathcal{A}^c)$. But for $n=4$, we see that 
	\begin{equation}
		\begin{aligned}
			&E(ABCD)
			+ E(AB|CD)+E(AC|BD)+E(AD|BC)\\
			&+E(BC|AD)+E(BD|AC)   +E(CD|AB) \\
			=\ &E(A|BCD)+E(B|ACD)+E(C|ABD)+E(D|ABC) \\
			&+E(ABC|D)+E(ABD|C)+E(ACD|B)+E(BCD|A).
		\end{aligned}
	\end{equation}
	Here we use $E(\mathcal{A}|\mathcal{A}^c)=S_L(\mathcal{A})=S_L(\Tr_{\mathcal{A}^c}(\rho_\psi))$ to denote the entanglement between $\mathcal{A}$ and $\mathcal{A}^c$.
	This provides a constraint of the pattern for entanglement of the state.
\end{example}

\begin{example}
	Let us now consider an example of $3$-qutrit state $\rho_{ABC}=|\psi\rangle \langle \psi|$, whose Bloch representation can be given explicitly using matrices in example \ref{exp:3d}: 
	\begin{equation}
		\rho_{ABC}=\sum_{\mu,\nu,\gamma=0}^{8} x_{\mu\nu \gamma} \sigma_{\mu} \otimes \sigma_{\nu}\otimes \sigma_{\gamma}. 
	\end{equation}
	From Eq.~(\ref{eq:R}), we see that 
	\begin{equation}
		(-1)^3 3^3 \Tr (R_{\psi}) =3^3 \Tr \rho_{\psi}^2 -2 (L_0+L_2).
	\end{equation}
	$L_0=1$ is a constant term. For $L_{2}=L_{2}^{AB}+L_2^{AC}+L_2^{BC}$, consider $L_{2}^{AB}=\sum_{\mu,\nu=1}^{8} x_{\mu\nu 0}^2$. It can be derived from reduced density matrix $\rho_{AB}$, $\rho_A$ and $\rho_B$, that
	\begin{equation}
		L_2^{AB}= 3^2 \Tr \rho_{AB}^2 - 3 \Tr \rho_A^2 -3\Tr \rho_B^2.
	\end{equation}
	This further implies that 
	\begin{equation}
		L_2^{AB}= 3^2(1-S_L(\rho_{AB}))- 3 (2-S_L(\rho_A)-S_L(\rho_B)).
	\end{equation}
	All other terms can be calculated similarly. In this way, we obtain all coefficients in Eq.~\eqref{eq:dis}, and thus the distribution formula of entanglement over the state $\psi$.
\end{example}

Notice that for more complicated cases, we could also use Lorentzian norms $\|\rho\|_{p,q}$ with a metric $\eta_{\mu,\nu}=(+,\cdots,+,-,\cdots,-)$. By assuming the invariance of these norms, we can obtain a formula that characterizes the distribution of entanglement over a multipartite state.
This means that these monogamy relations can be regarded as a result of the quantum geometric invariance.

\section{Application: antilinear superoperator symmetry of the open quantum system}
\label{sec:symmetry}

In this section, as an application of antilinear superoperators, let us consider the symmetries characterized by antilinear superoperators for open quantum systems.
Consider an open quantum system governed by Lindblad master equation  (also called quantum Liouville equation or GKSL equation) \cite{lindblad1976generators,gorini1976completely},
\begin{equation}\label{eq:Lindblad}
    \frac{d \rho}{d t}=\mathcal{L}(\rho)=\frac{1}{i\hbar}[H, \rho]+\sum_{i=1}^M\gamma_i( L_{i} \rho L_{i}^{\dagger}-\frac{1}{2}\{L_{i}^{\dagger} L_{i}, \rho\}),
\end{equation}
where $\mathcal{L}$ is a superoperator called Lindbladian, $H$ and $L_i$'s are Hamiltonian and Lindblad jump operators respectively, and $\gamma_i$'s are dissipation rate.
The study of symmetries of Lindbladian has attracted much attention in recent years \cite{buvca2012note,Albert2014symmetries,Lieu2020symmetry,Lieu2020tenfold,Altland2021symmetry,de2022symmetry,McDonald2022exact,Sa2023symmetry,Kawabata2023symmetry}. Antiunitary symmetry plays a crucial role in classifying Lindbladians. In this section, we will consider the most general case, where symmetries are characterized by superoperators.
These superoperator symmetries are generally not invertible, which may have potential applications in generalizing non-invertible symmetries \cite{cordova2022snowmass,brennan2023introduction,mcgreevy2023generalized,luo2023lecture,shao2024whats,SchaferNameki2024ICTP,Bhardwaj2024lecture} from closed systems to open systems.

\subsection{Superoperator symmetries: invertible and non-invertible}

For the closed quantum system with time evolution governed by the Schr\"{o}dinger equation, symmetry is an operator (or a collection of operators) that transforms the quantum state into the quantum state, and the expression of Schr\"{o}dinger equation remains unchanged (this means that the symmetry operator commutes with Hamiltonian). 
From Wigner's theorem, we know that there are two possible choices for such kind of symmetry operators: unitary and antiunitary operators.
Typical antiunitary symmetry operators are discrete symmetries, like time-reversal symmetry ($\mathsf{T}$), charge conjugation symmetry ($\mathsf{C}$) and composition of parity symmetry and time-reversal symmetry ($\mathsf{PT}$) \cite{peskin2018introduction,sachs1987physics,geru2019time,bender2019pt}.
The symmetry operator can take a more general form for an open quantum system.
Since the combination of an open system with its environment is a closed system, the unitary and antiunitary symmetry of this closed system results in CPTP and antilinear CPTP symmetries of the open quantum system. Formally, we can introduce the following definitions of superoperator symmetry for an open quantum system, following the approach in Refs.~\cite{buvca2012note,Albert2014symmetries,Lieu2020symmetry}.

\begin{definition}[Weak superoperator symmetry]
      A weak superoperator symmetry of an open quantum system governed by Lindblad master equation (as in Eq.~\eqref{eq:Lindblad}) is a linear or antilinear CPTP map $\mathcal{E}$ which is commutative with Lindbladian $\mathcal{L}$, i.e., 
      \begin{equation}\label{eq:StrongSym}
                \mathcal{L}\circ \mathcal{E}=\mathcal{E}\circ \mathcal{L}.
      \end{equation}
      If the Kraus rank (the minimal number of operators in the Kraus decomposition) of $\mathcal{E}$ is $k$, then it is called a rank-$k$ superoperator symmetry. The rank-1 CPTP superoperator symmetry corresponds to the unitary operator symmetry.
\end{definition}

Notice that the expression \eqref{eq:StrongSym} is different from that in Refs.~\cite{buvca2012note,Albert2014symmetries,Lieu2020symmetry}, where they use 
\begin{equation}\label{eq:unitarySym}
  \mathcal{E}\circ \mathcal{L}\circ \mathcal{E}^{\dagger}=\mathcal{L},  
\end{equation}
or 
\begin{equation}\label{eq:invertibleSym}
    \mathcal{E}\circ \mathcal{L}\circ \mathcal{E}^{-1}=\mathcal{L}.
\end{equation}
These expressions only work for unitary or invertible symmetries. Since we aim to make the Lindblad equation invariant under the symmetry, we will take Eq.~\eqref{eq:StrongSym} as a basic definition of strong superoperator symmetry. 
When $ \mathcal{E}$ is a unitary operator
\begin{equation}
    \mathcal{E}=U(\cdot)U^{\dagger}
\end{equation}
with $U$ a unitary operator. We have $ \mathcal{E}^{\dagger}=\mathcal{E}^{-1}$. Eq.~\eqref{eq:unitarySym} implies $\mathcal{E}\circ \mathcal{L}=\mathcal{L}\circ \mathcal{E}$.
For general non-unitary and more general non-invertible symmetry, the expressions  in Eqs.~\eqref{eq:unitarySym} and \eqref{eq:invertibleSym} do not work; instead, we should use expression~\eqref{eq:StrongSym} in the above definition to guarantee the invariance of Lindblad equation invariant under the symmetry. Under the superoperator symmetry, we have
\begin{equation}
    \frac{d \mathcal{E}( \rho)}{dt}=\mathcal{E}\circ \mathcal{L} (\rho)=\mathcal{L}\circ \mathcal{E}(\rho).
\end{equation}
Even when $\mathcal{E}$ is non-invertible, the Lindblad equation remains invariant. We have assumed that $\mathcal{E}$ is a CPTP map, which guarantees that $\mathcal{E}(\rho)$ is a density operator.  Also, notice that a CPTP map generally does not have a CPTP inverse, even when invertible.

Recall that for closed quantum systems, the symmetry is a unitary operator that commutes with the Hamiltonian. If we replace the unitary operator with a Hermitian operator, we obtain a physical observable $Q=Q(0)$; if this operator commutes with the Hamiltonian, then $Q$ is called a conserved quantity.
Notice that $[Q,H]=0$ implies that 
\begin{equation}
    \frac{d Q(t)}{dt} =0,
\end{equation}
where $Q(t)=\exp(it H) Q(0)\exp(-it H)$.
This is equivalent to
\begin{equation}
      \frac{d \langle \psi(t)| Q(0) |\psi(t)\rangle }{dt} =0,
\end{equation}
for all $\psi(t)$.
For open quantum systems, a similar notion can be introduced.
The time evolution determined by Lindbladian is 
\begin{equation}
    \mathcal{U}(t)=\exp(it \mathcal{L}).
\end{equation}
To ensure that 
\begin{equation}
    \frac{d}{dt} \Tr [ Q(0)\rho(t)] =0,
\end{equation}
we see this implies 
\begin{equation}
   \frac{d}{dt} Q(t)=  \frac{d}{dt} \exp(-it \mathcal{L^{\dagger}}) (Q(0)) =0.
\end{equation}
Notice that $\mathcal{L}$ is in general non-Hermitian. This means that the conserved quantity is defined as Hermitian operator $Q$ which satisfies~\cite{tarasov2008quantum,Albert2014symmetries}
\begin{equation}
    \mathcal{L}^{\dagger}(Q)=0.
\end{equation}
Notice that for closed quantum systems, symmetry and conserved quantity are related by Noether's theorem. When $Q$ is a conserved quantity, it generates a symmetry $V=\exp(i\alpha Q)$, which satisfies $[V,H]=0$.
The relationship between weak superoperator symmetry and conserved quantities remains an open problem to explore \cite{tarasov2008quantum,Albert2014symmetries}.
But there exists a special case called strong superoperator symmetry that can simplify the problem.

\begin{definition}[Strong superoperator symmetry]
\begin{enumerate}
    \item Let $\mathcal{E}$ be a linear rank-$k$ superoperator symmetry with Kraus operators $E_i$, $i=1,\cdots,k$. It is called a rank-$k$ strong linear superoperator symmetry if \cite{buvca2012note} 
      \begin{equation}\label{eq:ELcomm}
          [E_i,L_j]=[E_i,H]=[E_i,L_{j}^{\dagger}]=0,\quad  \forall\, i=1,\cdots,k, j=1,\cdots,M.
      \end{equation}
      The strong linear superoperator symmetry is automatically a weak linear superoperator symmetry but the converse is not true in general.

      \item The definition of antilinear strong symmetry $\mathcal{E}$ is more subtle. We cannot simply impose Eq.~\eqref{eq:ELcomm} on the Kraus operators, as this does not guarantee that $\mathcal{E}$ commutes with the Lindbladian $\mathcal{L}$.
      We will use the natural representations of $\mathcal{L}$ and $\mathcal{E}$ (see Lemma~\ref{lemma:relation}),
      \begin{gather}
          N_{\mathcal{L}}=-i(H\otimes \one -\one \otimes H^T)+\sum_{i>0} \gamma_i[L_i\otimes L^*_i-\frac{1}{2}(L^{\dagger}_iL_i\otimes \one +\one \otimes L_i^TL_i^{*})],\\
          N_{\mathcal{E}}=\sum_{i=1}^k (E_i\otimes E_i^{*} )\circ \mathcal{K}.
      \end{gather}
    An antilinear superoperator $\mathcal{E}$  is called an antilinear strong superoperator symmetry of $\mathcal{L}$ if and only if
    \begin{equation}
    \begin{aligned}
        [(E_i\otimes E_i^{*} )\circ \mathcal{K},-i(H\otimes \one )]= [(E_i\otimes E_i^{*} )\circ \mathcal{K},i(\one \otimes H^T)]=0, \quad \forall\, i, \\
         [(E_i\otimes E_i^{*})\circ \mathcal{K}, L_i\otimes L^*_i ]=0,\quad \forall\, i,j,\\
         [(E_i\otimes E_i^{*} )\circ \mathcal{K}, L^{\dagger}_iL_i\otimes \one]= [(E_i\otimes E_i^{*} )\circ \mathcal{K},\one \otimes L_i^TL_i^{*})]=0,\quad \forall\, i,j.
    \end{aligned}
    \end{equation}
    From the above definition, we see that the action of an antilinear strong superoperator symmetry on an open quantum system does not change the form of the Lindblad master equation. And antilinear strong superoperator symmetry is a special case of antilinear weak superoperator symmetry.
\end{enumerate}
\end{definition}

The Kraus-rank-1 linear superoperator symmetry of the open Heisenberg XXZ spin 1/2 chain is discussed in Ref.~\cite{buvca2012note}. However, the higher Kraus-rank linear symmetry and the antilinear superoperator symmetry remain largely unexplored \cite{buvca2012note,Albert2014symmetries,Lieu2020symmetry,Lieu2020tenfold,Altland2021symmetry,de2022symmetry,McDonald2022exact,Sa2023symmetry,Kawabata2023symmetry}.
A detailed discussion about the generalized symmetry for open quantum systems and their corresponding symmetry-protected phases will be provided in our future work~\cite{jia2024SPT}.

\subsection{Kramers' degeneracy}

A crucial property of antiunitary symmetry for a closed quantum system is the Kramers' degeneracy \cite{kramers1930theorie}.
The open quantum system also has this kind of degeneracy when considering steady states.
For a weak antiunitary  superoperator $\mathcal{M}$ (Definition~\ref{def:antiunitary}), when $\mathcal{M}^2=-\mathcal{I}$ with $\mathcal{I}$ the identity channel, we have
\begin{equation}\label{eq:zero}
    \langle \rho,\mathcal{M}(\rho)\rangle=0.
\end{equation}
To prove this, recalling that $\mathcal{M}^{\ddagger}=\mathcal{M}^{-1}$, we have
\begin{equation}
    \langle \rho, \mathcal{M}(\rho)\rangle 
    = -\langle \mathcal{M}^2(\rho) ,\mathcal{M}(\rho)
    =-\langle \mathcal{M}^{\ddagger}\circ \mathcal{M} \circ \mathcal{M}(\rho),\rho\rangle^*
    =-\langle \rho ,\mathcal{M}(\rho)\rangle,
\end{equation}
which implies Eq.~\eqref{eq:zero}.

\begin{theorem}[Kramers' degeneracy for steady states of open system]
    Consider an open system with Lindbladian $\mathcal{L}$. If there is a weak superoperator symmetry characterized by a weak antiunitary CPTP map $\mathcal{M}$ which satisfies $\mathcal{M}^2=-\mathcal{I}$, then the steady state space must be degenerate. 
\end{theorem}

\begin{proof}
Notice that $\mathcal{M}$ is an antilinear CPTP map, meaning $\mathcal{M}(\rho)$ is a density operator if $\rho$ is a density operator.
 For the steady state $\rho$, we have
 \begin{equation}
     \mathcal{L}(\rho)=0.
 \end{equation}
 Since $\mathcal{M}$ is a weak symmetry of $\mathcal{L}$, we obtain
 \begin{equation}
     \mathcal{L}\circ \mathcal{M}(\rho)=\mathcal{M}\circ \mathcal{L}(\rho)=0.
 \end{equation}
 This implies that $\mathcal{M}(\rho)$ is also a steady state.
 We have shown that when $\mathcal{M}^2=-\mathcal{I}$,  $\mathcal{M}(\rho)$ is orthogonal to $\rho$. Thus, $\mathcal{M}(\rho)$ and $\rho$ are independent steady states, leading us to our conclusion.
\end{proof}

\section{Conclusion}
\label{sec:conclusion}

In this work, we systematically investigate the antilinear superoperators, including various representations and properties of antilinear quantum channels, antilinearly unital superoperators, antiunitary superoperators, and generalized $\Theta$-conjugations.
The generalized $\Theta$-conjugation plays an important role in studying the geometric properties of quantum states.
Using Bloch representation of a higher-dimensional quantum system, different generalized $\Theta$-conjugations provide different metrics of the space of Bloch space-time vectors; both the Lorentzian and Euclidean invariance of quantum states can be investigated in this framework.
The invariant class is just the states with the same norms in the corresponding space.
Using these geometric properties, we derive the monogamy equalities of entanglement from this geometric invariance for many-body quantum states.
This means that distribution of entanglement over a multipartite state can be regarded as a result of geometric invariance.
The application of an antilinear superoperator in characterizing the symmetry of an open quantum system is also briefly discussed. Time reversal symmetry and $\mathsf{PT}$-symmetry are typical examples.

The generalized $\Theta$-conjugated fidelity and concurrence are also introduced. 
Though we did not discuss more details about $\Theta$-conjugated fidelity in this work, we would like to point out that the quality is closely related to the antilinear superoperator.
Notice that for pure state $F_1(\psi,\varphi) =|\langle \psi|\varphi\rangle|$; Wigner's theorem claims that quantum operations that preserve this fidelity are unitary and antiunitary ones. It is natural to ask what is the quantum operations that preserve the generalized $p$-norm $\Theta$-conjugated fidelity. This interesting open problem will be left for our future study.
For the qubit case, the generalized $\Theta$-conjugated concurrence is related to the norm of Bloch vectors; however, for higher-dimensional case, there is no such correspondence. 

The framework is also useful for investigating various quantum correlations, including Bell nonlocality, quantum steering, and quantum entanglement. Especially for the two-particle case, using the Bloch vectors, we can obtain geometric bodies corresponding to these correlations and study their relations with antilinear superoperators. This part is also left for our future study.

We also discussed the antilinear superoperator symmetries of open quantum systems, which may have applications in understanding generalized symmetries and symmetry-protected topological phases for mixed states. A detailed discussion will be provided in our future work \cite{jia2024SPT}.

\acknowledgements
 
\emph{Z.J. is the corresponding author of this work.
Z.J. and D.K. are supported by the National Research Foundation in Singapore and A*STAR under its CQT Bridging Grant.
L.W. acknowledges Nelly Ng for discussions during her visit to NTU. S.T. would like to thank Professor Uli Walther for his constant support and encouragement.
S.T. was partly supported by NSF grant DMS-2100288 and Simons Foundation Collaboration Grant for Mathematicians \#580839, and is now supported by funding from Quantum Symmetry Group of BIMSA. All authors are grateful for the referee’s valuable suggestions. }

\appendix

\section{More on antilinear superoperators}
\label{app:AntiSup}
In this appendix, we collect some other properties of antilinear superoperators. The Kraus-rank-1 case has been extensively investigated in Ref.~\cite{uhlmann2016anti}. Here we present some generalizations about antilinear superoperators. Two of the main changes are that $\Theta$-conjugation becomes generalized $\Theta$-conjugation and antiunitary operator becomes weak antiunitary superoperators.

Given a Hilbert space $\mathcal{X}$, we define $\mathbf{B}^{(n)}(\mathcal{X})$ recursively as follows: $\mathbf{B}^{(0)}(\mathcal{X})=\mathcal{X}$, and $\mathbf{B}^{(n)}(\mathcal{X})$ is the set of all linear maps from $\mathbf{B}^{(n-1)}(\mathcal{X})$ to $\mathbf{B}^{(n-1)}(\mathcal{X})$.
The $n$-th order transformation between two quantum systems $\mathcal{X}$ and $\mathcal{Y}$ is a map between $\mathbf{B}^{(n-1)}(\mathcal{X})$ and $\mathbf{B}^{(n-1)}(\mathcal{Y})$; the set of all $n$-th order transformations is denoted as $\mathbf{B}^{(n)}(\mathcal{X},\mathcal{Y})$. 
In this sense, a quantum channel is a $2$nd order transformation.
For antilinear case, we have similar definition: $\mathbf{B}^{(n)}_{\rm anti}(\mathcal{X},\mathcal{Y})$ consists of all antilinear maps from $\mathbf{B}^{(n-1)}(\mathcal{X})$ to $\mathbf{B}^{(n-1)}(\mathcal{Y})$.
In quantum information theory, we are mainly interested in $\mathbf{B}^{(0)}(\mathcal{X})$ (state vectors), $\mathbf{B}^{(1)}(\mathcal{X},\mathcal{Y})$ and  $\mathbf{B}^{(1)}_{\rm anti}(\mathcal{X},\mathcal{Y})$ (density operators and transformations of state vectors),  $\mathbf{B}^{(2)}(\mathcal{X},\mathcal{Y})$ and  $\mathbf{B}^{(2)}_{\rm anti}(\mathcal{X},\mathcal{Y})$ (quantum operations).

If an $n$-th order antilinear transformation $\mathcal{M} \in \mathbf{B}^{(n)}_{\rm anti}(\mathcal{X})$ is bijective, then it is invertible and its inverse is also antilinear. 
The spectrum of an antilinear transformation (whenever exists) consists of a collection of concentric circles with $0$ as their common center in the complex plane. Namely, if $\lambda$ is an eigenvalue of $\mathcal{M}$, then for arbitrary $\alpha\in \mathbb{R}$, $e^{i\alpha} \lambda$ is also an eigenvalue.
The eigenstates of an $n$-th order antilinear transformation are  $(n-1)$-th order linear transformations, e.g., the eigenstates of superoperators are operators (called eigen-operators). If an antilinear transformation has eigenvalues and eigenstates, it is called diagonalizable.
For any diagonalizable $\mathcal{M}$, the eigenvalues of $\mathcal{M}^2$ are non-negative real numbers.
A crucial extreme example is the time reversal operator $\mathsf{T}$, for which $\mathsf{T}^2=\pm \mathcal{I}$ \cite{peskin2018introduction,sachs1987physics,geru2019time,bender2019pt}. When $\mathsf{T}^2=-\mathcal{I}$, $\mathsf{T}$ must not be diagonalizable.

An antilinear superoperator $\mathcal{M}$ is called Hermitian if $ \mathcal{M}^{\ddagger}=\mathcal{M}$, and skew Hermitian if $\mathcal{M}^{\ddagger}=-\mathcal{M}$.
Recall that natural representation preserves the Hermitian adjoint $N(\mathcal{M}^{\ddag})=N(\mathcal{M})^{\ddag}$, hence it is easy to verify that
the linearization of natural representation of an antilinear Hermitian (resp. skew Hermitian) superoperator must be symmetric (resp. skew-symmetric).

\begin{theorem}
   If an antilinear superoperator $\mathcal{M}$ is Hermitian, then there exists a basis of eigen-operators. If antilinear $\mathcal{M}$ is skew Hermitian, then it has no eigen-operators.
\end{theorem}

\begin{proof}
    Fix the bases of underlying Hilbert spaces, then the assertions followed by using the natural representation and lemma~3.4 in Ref.~\cite{uhlmann2016anti}.
\end{proof}

It can also be proved that antilinear $\mathcal{M}$ is diagonalizable if and only if $\mathcal{M}$ is Hermitian in some given inner product.
Any antilinear superoperator can be decomposed into a summation of an antilinear Hermitian superoperator and an antilinear skew Hermitian superoperator.

For antilinear superoperator $\mathcal{M}$, notice that $\langle \rho ,\mathcal{M}^{\ddag}\circ \mathcal{M}(\rho)\rangle=\langle \mathcal{M}(\rho),\mathcal{M}(\rho)\rangle \geq 0$ for all $\rho$, which implies that  $\mathcal{M}^{\ddag} \circ \mathcal{M}\geq 0$. Similarly $\mathcal{M}\circ \mathcal{M}^{\ddag}\geq 0$.
We then can introduce the linear superoperators $|\mathcal{M}|_l=\sqrt{\mathcal{M}^{\ddag} \circ \mathcal{M}}$ and $|\mathcal{M}|_r=\sqrt{\mathcal{M}\circ \mathcal{M}^{\ddag}}$, whose natural representations are both positive semidefinite operators.
 
\begin{theorem}[Polar decomposition]
Let $\mathcal{M}$ be an antilinear superoperator. There exists a weak antiunitary superoperator $\mathcal{U}$ such that
\begin{equation}
    \mathcal{M}=\mathcal{U}\circ |\mathcal{M}|_l=|\mathcal{M}|_r \circ \mathcal{U}.
\end{equation}
\end{theorem}

\begin{proof}
Suppose $\rho_i$'s are a basis of eigen-operators of $|\mathcal{M}|_l$ with eigenvalues $\lambda_i$. Let $\tilde{\chi}_i=\mathcal{M}(\rho_i)$, then
$\langle \tilde{\chi}_i,\tilde{\chi}_i\rangle= \lambda_i^2$.
Define $\chi_i=\tilde{\chi}_i/\lambda_i$, and $\mathcal{U}: \rho_i\mapsto \chi_i$ (antilinearly extend to the whole space).
It is easy to verify that $\mathcal{M}(\rho)=\mathcal{U}\circ |\mathcal{M}|_l(\rho)$ for all $\rho$.

Then using $\mathcal{M}=\mathcal{U}\circ |\mathcal{M}|_l \circ \mathcal{U}^{\ddag}\circ \mathcal{U}$,
since $(\mathcal{U}\circ |\mathcal{M}|_l \circ \mathcal{U}^{\ddag})\circ (\mathcal{U}\circ |\mathcal{M}|_l \circ \mathcal{U}^{\ddag})^{\ddag}= \mathcal{M}\circ \mathcal{M}^{\ddag}$, we see $\mathcal{U}\circ |\mathcal{M}|_l \circ \mathcal{U}^{\ddag}=|\mathcal{M}|_r$.
This completes the proof.
\end{proof}


Antilinear superoperator $\mathcal{M}$ is called an involution if $\mathcal{M}^2=\mathcal{I}$; it is called a skew involution if $\mathcal{M}^2=-\mathcal{I}$.
If $\mathcal{M}$ is an involution (resp. skew involution), then $\mathcal{M}^{\ddag}$ is also an involution (resp. skew involution).
Since for involution or skew involution, we have $[\mathcal{M}\mathcal{M}^{\ddagger},\mathcal{M}^{\ddagger}\mathcal{M}]=0$, it follows that $\sqrt{\mathcal{M}\mathcal{M}^{\ddagger}}\sqrt{\mathcal{M}^{\ddagger}\mathcal{M}}=\sqrt{\mathcal{M}\mathcal{M}^{\ddagger}\mathcal{M}^{\ddagger}\mathcal{M}}=\mathcal{I}$. This implies that $|\mathcal{M}|_r^{-1}=|\mathcal{M}|_l=|\mathcal{M}^{\ddag}|_r$.

\begin{theorem}\label{thm:involution}
   Let $\mathcal{M}$ be an antilinear superoperator. The following hold:
   \begin{enumerate}
       \item If $\mathcal{M}$ is an involution, then we have the polar decomposition 
       \begin{equation} \label{eq:Minvo}
      \mathcal{M}=|\mathcal{M}|_r\circ \mathcal{U}=\mathcal{U}^{\ddag}\circ |\mathcal{M}|_r^{-1}.
       \end{equation}
     Then $\mathcal{U}$ must be a generalized $\Theta$-conjugation.
     \item If $\mathcal{M}$ is a skew involution, then we have the polar decomposition 
       \begin{equation}
         \mathcal{M}=|\mathcal{M}|_r\circ \mathcal{U}= -\mathcal{U}^{\ddag} \circ |\mathcal{M}|_r^{-1}.
       \end{equation}
     Then $\mathcal{U}$ must be a generalized skew $\Theta$-conjugation.
   \end{enumerate}
\end{theorem}

\begin{proof}
  1.  Notice that $\mathcal{M}=\mathcal{M}^{-1}$. By taking inverse of the right polar decomposition $\mathcal{M}=|\mathcal{M}|_r\circ \mathcal{U}$ and notice that $\mathcal{U}$ is antiunitary, we obtain expression \eqref{eq:Minvo}. Then compare it with left polar decomposition, we have
  $\mathcal{U}=\mathcal{U}^{\ddag}=\mathcal{U}^{-1}$.
  
  2. Similar to the proof of statement 1.
\end{proof}

\begin{theorem}
   Let $\mathcal{M}\in \mathbf{B}^{(2)}(\mathcal{X})$ be an involution or a skew involution, then $\mathcal{M}$ is a generalized $\Theta$-conjugation or generalized skew $\Theta$-conjugation if and only if 
   \begin{equation}
       \Tr (|\mathcal{M}|_r)=\Tr (|\mathcal{M}|_l)=\dim \mathbf{B}(\mathcal{X}).
   \end{equation}
   This further implies that an involution $\mathcal{M}$ is a generalized $\Theta$-conjugation if and only if it is normal $[\mathcal{M},\mathcal{M}^{\ddag}]=0$.
\end{theorem}

\begin{proof}
    This is a direct corollary of theorem~\ref{thm:involution}.
\end{proof}

A discrete antilinear quantum instrument $\mathfrak{I}$ is a family of antilinear CP superoperators $\mathcal{M}_i, i\in I$ (with $I$ discrete) such that $\mathcal{M}_{\mathfrak{I}}=\sum_{i\in I} \mathcal{M}_i$ (which is automatically antilinearly CP) is antilinearly TP.

Many crucial classes of linear quantum instruments can be generalized to antilinear case:
\begin{itemize}
    \item Antlinear channel: $\mathbf{TPCP}_{\rm anti}(\mathcal{X},\mathcal{Y})$;
    \item Antilinear separable map: $\mathbf{SEP}_{\rm anti}(\mathcal{X},\mathcal{Y})$;
    \item Antilinear local operation and classical communication: $\mathbf{LOCC}_{\rm anti}(\mathcal{X},\mathcal{Y})$;
     \item Antilinear stochastical local operation and classical communication: $\mathbf{SLOCC}_{\rm anti}(\mathcal{X},\mathcal{Y})$.
\end{itemize}

Using linearization theorem, the distance measure of linear superoperators can be applied to antilinear superoperators.
\begin{equation}
    D_{\rm anti}(\mathcal{M},\mathcal{N}):=D(\mathcal{M}_L,\mathcal{N}_L).
\end{equation}
In this way, we can define operator norm, diamond norm and so on.

\section{Properties of $p$-norm fidelity and $p$-norm concurrence}
\label{App:FidelityP}

Many crucial properties of fidelity (resp. concurrence) can be generalized to $p$-norm fidelity (resp. $p$-norm concurrence). 
\begin{theorem}
    $F_p(\rho,\sigma)$ ($p\geq 1$) fulfill the following properties:
    \begin{enumerate}
        \item Monotone decreasing in $p$: $F_p(\rho,\sigma)\leq F_q(\rho,\sigma)$ for $p\geq q\geq 1$; for pure states $F_{p}(\psi,\varphi)=F_1(\psi,\varphi)$ for all $p$.
        
        \item Symmetric: $F_p(\rho,\sigma)=F_p(\sigma,\rho)$ for any $\rho, \sigma$. 
        
        \item Bounded: $0\leq F_p(\rho,\sigma)\leq 1$ and $F_p(\rho,\sigma)=1$ only if $\rho=\sigma$. 
        
        \item Unitary and antiunitary invariance: $F_p(\rho,\sigma)=F_p(U\rho U^\dagger,U\sigma U^\dagger)$ for any (anti)unitary operator $U$. 
        
        \item For density operators $\rho_1,\rho_2$ and $\sigma_1,\sigma_2$, we have $F_p(\rho_1\otimes \rho_2,\sigma_1\otimes \sigma_2)=F_p(\rho_1,\sigma_1)F_p(\rho_2,\sigma_2)$.
    
        \item Homogeneous: $F_p(\lambda \rho, \lambda \sigma)=\lambda F_p(\rho,\sigma)$ for $\lambda \in \mathbb{R}_{\geq 0}$.
    \end{enumerate}  
\end{theorem}

\begin{proof}
    1. Recall the algorithm to compute $F_p(\rho,\sigma)$ (in our case): first compute the singular values of $\sqrt{\rho}\sqrt{\sigma}$, say $\lambda_1,\cdots,\lambda_d$ (which are all non-negative numbers), then $F_p(\rho,\sigma)$ is equal to the $p$-norm of the vector $\vec{x}:=({\lambda_1},\cdots,{\lambda_d})$. In fact, for any vector $\vec{x}\in\mathbf{R}^d$, the $p$-norm $\|\vec{x}\|_p:=(\sum_{i=1}^dx_i^p)^{1/p}$ is monotone decreasing in $p$. This is well known in the literature. We may assume that $x_i>0$ which is enough for our case. Consider the function $g(p)=\|\vec{x}\|_p=\exp(\frac{1}{p}\log (\sum_{i=1}^d x_i^p))$ for $p\geq 1$. An easy computation shows that $g'(p)=\frac{1}{p}h(p)\|\vec{x}\|_p^{1-p}$, where $h(p)=\sum_{i=1}^d x_i^p\log x_i -\|\vec{x}\|_p^p\log\|\vec{x}\|_p$. As $x_i\leq \|\vec{x}\|_p$ for each $i$, taking $\log$, multiplying $x_i^p$ and taking sum on both sides yields $\sum_{i=1}^dx_i^p\log x_i\leq \|\vec{x}\|_p^p\log \|\vec{x}\|_p$. Thus $h(p)\leq0$, showing $g'(p)\leq0$ and hence $g(p)=\|\vec{x}\|_p$ is monotone decreasing in $p$. 
    
    2. This is obvious since $\sqrt{\rho}\sqrt{\sigma}$ and $\sqrt{\sigma}\sqrt{\rho}$ have the same nonzero singular values. 
    
    3. $F_p(\rho,\sigma)\geq0$ is obvious. By part 1 and the property of 1-fidelity, we have $F_p(\rho,\sigma)\leq F_1(\rho,\sigma)\leq 1$. Moreover, $F_p(\rho,\sigma)=1$ forces $F_1(\rho,\sigma)=1$, which in turn implies $\rho=\sigma$. 
    
    4. The case for unitary $U$ follows from the fact that $\sqrt{U\rho U^\dagger}=U\sqrt{\rho}U^\dagger$. The case for antiunitary $U$ follows from theorem~\ref{lemma:antilinear} and the fact that $F_p(\rho^*,\sigma^*)=F_p(\rho,\sigma)$, which is true since the singular values of $\sqrt{\rho^*}\sqrt{\sigma^*}$ and $\sqrt{\rho}\sqrt{\sigma}$ are conjugate (hence the same in our case because they are real numbers). 

    5. This is clear from $\| \sqrt{\rho_1\otimes \rho_2} \sqrt{\sigma_1\otimes \sigma_2}\|_p=\|\sqrt{\rho_1} \sqrt{\sigma_1}\|_p \| \sqrt{\rho_2} \sqrt{\sigma_2}\|_p$.

   6. This is clear from the definition.
\end{proof}

\begin{theorem}
   For $p$-concurrence $C_p(\rho,\sigma)$, the following statements hold:  
   \begin{enumerate}
       \item Symmetric: $C_p(\rho,\sigma)=C_p(\sigma,\rho)$ for any $\rho, \sigma$. 
       \item Bounded: $0\leq C_p(\rho,\sigma)\leq 1$.
       \item Homogeneous: $C_p(\lambda \rho, \lambda \sigma) =\lambda C_p(\rho,\sigma)$ for $\lambda \in \mathbb{R}_{\geq 0}$.
   \end{enumerate}
\end{theorem}

\begin{proof}
    1. This is clear since $\sqrt{\rho}\sqrt{\sigma}$ and $\sqrt{\sigma}\sqrt{\rho}$ have the same nonzero singular values. 
    
    2. This is clear from the fact that, for all singular values $
    \lambda_i$ of $\sqrt{\rho}\sqrt{\sigma}$, $0\leq \sum_i \lambda_i^p \leq 1$.
    
    3. This is clear from the definition.
\end{proof}

\section*{Declarations}
\paragraph{Conflict of interest:} All authors certify that there are no conflicts of interest for this work.

\paragraph{Data Availability Statement:} Data sharing is not applicable to this article as no datasets were generated or analyzed during the current study.

\bibliographystyle{apsrev4-1-title}
\bibliography{mybib}
	
\end{document}